\documentclass[11pt, reqno]{amsart}
\usepackage{graphicx,amsmath,amssymb,epsfig,color}
\usepackage{amsfonts,cleveref,float,subfig,rotating}
\usepackage{appendix}
\usepackage[round]{natbib}

\newtheorem{theorem}{Theorem}
\theoremstyle{plain}

\newtheorem{corollary}{Corollary}

\newtheorem{definition}{Definition}

\newtheorem{lemma}{Lemma}

\newtheorem{proposition}{Proposition}

\theoremstyle{definition}
\newtheorem{assump}{Assumption}
\newenvironment{myassump}[2][]
  {\renewcommand\theassump{#2}\begin{assump}[#1]}
  {\end{assump}}

\newtheorem*{definition*}{Definition}

\begin{document}
\title[Identification of hedonic equilibrium]{Identification of hedonic equilibrium and nonseparable simultaneous equations}
\author{Victor Chernozhukov, Alfred Galichon, Marc Henry and Brendan Pass}
\address{MIT, NYU and Sciences-Po, Penn State, University of Alberta}
\date{{\scriptsize {First version: 9 February 2014. The present version is of \today. The authors thank
Guillaume Carlier, Andrew Chesher, Pierre-Andr\'e Chiappori, Hidehiko Ichimura, Arthur Lewbel, Rosa Matzkin and multiple seminar audiences
for stimulating discussions. They also thank the editor James Heckman and four referees for insightful comments. Chernozhukov's research has received funding from the NSF. Galichon’s research has received funding from ERC grants CoG-866274 and 313699, from NSF grant DMS-1716489, and from FiME (Laboratoire de Finance des Marchés de l’Energie). Henry's research has received
funding from SSHRC Grants 435-2013-0292 and NSERC Grant
356491-2013. Pass' research has received support from a University of Alberta start-up grant and NSERC grants 412779-2012 and 04658-2018.}}}

\begin{abstract}
This paper derives conditions under which preferences and technology are
nonparametrically identified in hedonic equilibrium models, where products
are differentiated along more than one dimension and agents are
characterized by several dimensions of unobserved heterogeneity. With
products differentiated along a quality index and agents characterized by
scalar unobserved heterogeneity, single crossing conditions on preferences
and technology provide identifying restrictions in \citet{EHN:2004} and \citet{HMN:2010}. We develop similar shape
restrictions in the multi-attribute case. These shape restrictions, which are
based on optimal transport theory and generalized convexity, allow us to identify
preferences for goods differentiated along multiple dimensions, from the observation of a single market.
We thereby derive identification results for nonseparable simultaneous equations and multi-attribute hedonic equilibrium models with (possibly) multiple dimensions of unobserved heterogeneity.
One of our results is a proof of absolute continuity of the distribution of endogenously traded qualities,
which is of independent interest.
\end{abstract}

\maketitle
\thispagestyle{empty}

\noindent

{\footnotesize
\textbf{Keywords}: Hedonic equilibrium, multivariate quantile
identification, multidimensional unobserved heterogeneity, cyclical monotonicity, optimal
transport.

\textbf{JEL subject classification}: C14, C39, C61, C78
}


\section*{\protect Introduction}

Hedonic models were
initially introduced by \citet{Court:39} to price highly
differentiated goods in terms of their attributes. The vast subsequent literature on hedonic regressions
of prices on attributes aimed at measuring the marginal willingness of consumers to pay for the attributes of the good they acquired, or the marginal willingness of workers to accept compensation for the attributes of their occupations. When unobservable taste for attributes drives the consumers' choices, however, a simple regression of price on attributes cannot inform us on the willingness to pay for different quality levels from the ones characterizing the good actually acquired. Nor can they inform us on the willingness to pay for characteristics of the good they would acquire under counterfactual market conditions, with different endowments, preferences and technology.

The willingness to pay for counterfactual transactions, together with structural parameters of preferences and technology, can be recovered with a general equilibrium theory of hedonic models, dating back to \citet{Tinbergen:1956} and \citet{Rosen:1974} (see \citet{Heckman:2019} for an account of their respective contributions).
The common underlying framework, which we also adopt here, is that of a
perfectly competitive market with heterogeneous buyers and sellers and
traded product quality bundles and prices that arise endogenously in
equilibrium\footnote{When preferences are quasi-linear in price and under mild
semicontinuity assumptions, \citet{Ekeland:2010} and %
\citet{CMN:2010} show that equilibria exist, in the form of a joint
distribution of product and consumer types (who consumes what), a joint
distribution of product and producer types (who produces what) and a price
schedule such that markets clear for each traded product.}.
\citet{Rosen:1974} proposes a
two-step procedure to estimate general hedonic models and thereby analyze
general equilibrium effects of changes in buyer-seller compositions,
preferences and technology on qualities traded at equilibrium and their
price (see \citet{Heckman:99}). The first step is a regression of prices on attributes, and the second is a simultaneous equations estimation of the demand and supply system with marginal prices estimated in the first step as endogenous variables.

\citet{Brown:83} and \citet{BR:1982} point out that changes in consumers' unobserved taste for attributes would lead them to source goods from different suppliers, so that exclusion restrictions from the supply side cannot be justified\footnote{See also \citet{Epple:1987}, \citet{Bartik:1987} and \citet{KL:1988}.}. The literature on recovering marginal willingness to pay for counterfactual transactions has since followed three strategies: relying on multiple markets across space or time (see \citet{Brown:83}, \citet{BR:1982}, \citet{KL:1988}, \citet{TW:2001} and many references in \citet{KST:2013}); relying on specific functional forms for utility (\citet{BB:2005}, \citeauthor{BT:2011} \citeyear{BT:2011,BT:2019} and references therein); or assume consumers care about a single dimension of good heterogeneity, via a quality index (\citeauthor{EHN:2004} \citeyear{EHN:2002,EHN:2002a,EHN:2004}, \citeauthor{HMN:2010} \citeyear{HMN:2003,HMN:2010}, \citet{EPS:2010} and \citet{EQS:2020}). Each of these strategies has drawbacks. Multimarket strategies rely on the assumption of no leakage between markets, or consumption substitution across time and space. They also rely on the assumption that preferences and the distribution of preference types are stable across markets, so that variation comes from the supply side and preferences are identified, but not technology (or vice versa with symmetric assumptions, see \citet{EHN:2004}). Identification strategies based on specific parameterizations of preferences and technology cannot distinguish features of the specification that are crucial to identification and features that are convenient approximations. Identification proofs must also be repeated for each new parameterization, the suitability of which depends on the application (see the discussion in \citet{Yinger:2014}). Finally, it is important in many applications to account for heterogeneity in consumers' (or workers') relative valuations of different attributes, and hence move beyond the case of a scalar index of attributes.

This paper proposes an identification strategy based on a single market, where agents have heterogeneous relative valuations for different attributes, without relying on a specific parametric specification of preferences and technology. A leading case in the class of specifications we entertain is~$U(x,\varepsilon,z)=\bar U(x,z)+z'\varepsilon,$ where~$U$ is the valuation of the bundle of attributes~$z$ as a function of the vectors of observable and unobservable consumer characteristics~$x$ and~$\varepsilon$ respectively. We show that for each choice of distribution of types~$\varepsilon$, the function~$\bar U$ is recovered nonparametrically (and the same result holds for the supply side).
Our contribution is a direct generalization of the main identification strategy in \citet{HMN:2010}. In the latter, under a single crossing condition\footnote{Also known as Spence-Mirlees or supermodularity condition.} on the utility
function, the first order condition of the consumer problem yields an increasing
demand function, i.e., quality demanded by the consumer as an increasing
function of her unobserved type, interpreted as unobserved taste for
quality. Assortative matching guarantees uniqueness of demand, as the unique
increasing function that maps the distribution of unobserved taste for
quality, which is specified a priori, and the distribution of qualities,
which is observed. Hence demand is identified as a quantile function, as in %
\citet{Matzkin:2003}. Identification, therefore, is driven by a shape
restriction on the utility function.

 The main achievement of this paper is to show that a suitable multivariate extension (called {\em twist condition})
 of the single crossing shape restriction delivers the same identification result
 in hedonic equilibrium with multiple good quality dimensions. Heuristically, the proof mirrors \citet{HMN:2010}
 in that is first involves showing identification of inverse demand, which then allows the identification of marginal utility
 from the first order condition of the consumer's program. The identification of inverse demand, i.e., a single valued mapping from
 a vector of good qualities to a vector of unobserved consumer type, involves the twist shape restriction and cyclical monotonicity
 of the hedonic equilibrium solution. The recovery of marginal utility from the first order condition of the consumer's program
 is involved, because differentiability of the hedonic equilibrium price function is not guaranteed. Known conditions for
 differentiability of transport potentials in general optimal transport problems, due to \citet{MTW:2005} and which would yield differentiability of the hedonic equilibrium price in our context, are very strong and rule out many simple forms of the matching surplus (see Chapter~12 of \citet{Villani:2009}). We are able to by-pass the \citet{MTW:2005} conditions using the special structure of the hedonic equilibrium, and
to show approximate differentiability (Definition~10.2 page 218 of \citet{Villani:2009}) of the price function,
 for which we need absolute continuity of the distribution of good qualities traded at equilibrium. To that end, we provide a set of mild conditions on the primitives under which the endogenous distribution of qualities traded at equilibrium is absolutely continuous.
The proof of absolute continuity of the distribution of qualities traded at equilibrium is based on an argument from \citet{FJ:2008}, also applied in \citet{KP:2014}\footnote{\citet{FJ:2008} and \citet{KP:2014} focus on quadratic distance cost, i.e.,~$\zeta(x,\varepsilon,z)=-d^2(z,\varepsilon)$ in the notation of Assumption~\ref{ass:uh}, and work in more exotic geometric spaces.}.

 An important special case of our main identification theorem is the case, where the consumer's utility depends on consumer unobserved heterogeneity~$\varepsilon$ only
 through the index~$z'\varepsilon$, where~$z$ is the vector of good qualities. This case has the appealing interpretation that each dimension of unobservable taste is associated with a
 good quality dimension and the appealing feature that marginal utility is characterized as the solution of a convex program. However, choosing the dimension of unobserved consumer heterogeneity to be equal to the dimension of the vector of good qualities is a somewhat arbitrary modelling choice, and we provide an extension of our main identification theorem to cases where the dimension of unobserved consumer heterogeneity is lower, including a model with scalar unobserved heterogeneity. We derive a local identification result under mild conditions, but for global identification, we need a shape restriction on the endogenous price function, for which we know of no sufficient conditions on primitives. Another restrictive aspect of our main result is the necessary normalization of the distribution of unobserved heterogeneity, when identifying primitives from a single market. We provide some relaxation of this constraint when data from multiple  markets is available, but the results are still fragmentary.

The analysis of identification of inverse demand in hedonic equilibrium reveals that inverse demand satisfies a multivariate notion of monotonicity, whose definition depends on the form of the utility function that is maximized. In the univariate case, this notion reduces to monotonicity of inverse demand. In the case, where the consumer's utility depends on consumer unobserved heterogeneity~$\varepsilon$ only
 through the index~$z'\varepsilon$, where~$z$ is the vector of good qualities, inverse demand is the gradient of a convex function. We show that this notion of multivariate monotonicity is a suitable shape restriction to identify nonseparable simultaneous equations models, generalizing the quantile identification method of \citet{Matzkin:2003} and complementing results in \citet{Matzkin:2015}, where monotonicity is imposed equation by equation\footnote{Not all results in \citet{Matzkin:2003} and \citet{HMN:2010} require normalization of the distribution of unobserved heterogeneity, as we do here.}.

\subsubsection*{\protect Closely related work}

On the identification of multi-attribute hedonic equilibrium models, \citet{EHN:2004} require marginal utility (resp. marginal product) to
be additively separable in unobserved consumer (resp. producer)
characteristic. \citet{HMN:2010} show that demand is nonparametrically
identified under a single crossing condition and that various additional
shape restrictions allow identification of preferences without additive
separability (see also \citeauthor{HMN:2003} \citeyear{HMN:2003,HMN:2005}).
\citet{EHN:2004} emphasize the one-dimensional case but argue that their results can
be extended to multivariate attributes under a separability assumption in
the utility (without restricting the distribution of unobserved heterogeneity).
Our paper directly  follows \citet{HMN:2010} and generalizes the insight therein to allow for heterogeneity in response to different dimensions of amenities or qualitites. However, the analyses in \citet{HMN:2010} and ours are non nested. \citet{HMN:2010} considers several specifications with Barten scales that are outside the scope of our generalization.
\citet{Nesheim:2013} (developed independently and concurrently) is the most closely related paper and complements our work. He achieves identification under an additive separability restriction,
but without restricting the distribution of unobserved heterogeneity. He also imposes conditions from \citet{MTW:2005} to obtain differentiability of the price.
\citet{CMN:2010} derive a matching
formulation of hedonic models and thereby highlight the close relation
between empirical strategies in matching markets and in hedonic markets.
\citet{CMP:2016} have a section on identification (posterior to this paper), where they
use a similar strategy. Two main differences arise from the difference between matching and hedonic models. On the one hand, \citet{CMP:2016} need an extra step to account for the fact that in their matching model the price function is not observed. On the other hand, they do not have to worry about regularity of endogenous objects such as the price and distribution of goods in the hedonic equilibrium setting.
\citet{GS:2012} extend the work of \citet{CS:2006} and identify preferences in marriage markets, where agents match on discrete characteristics, as the unique solution of an optimal transport problem, but unlike the present paper, they are restricted to the case with a discrete quality space. The strategy is extended to the set-valued case by \citet{CGS:2013}, who use subdifferential calculus to identify  dynamic discrete choice problems. \citet{DGH:2014} use network flow techniques to identify discrete hedonic models.

On the identification of nonlinear simultaneous equations models, \citet{Matzkin:2015} uses equation by equation monotonicity in the one dimensional unobservables and exclusion restrictions.  \citet{BH:2018} consider transformation models, as in \citet{Matzkin:2008} (see also \citet{Matzkin:2013}). These strategies do not require normalization of the distribution of unobserved heterogeneity.
\citet{SSS:2018} also (independently and in a very different context) use cyclical monotonicity for identification in panel discrete choice models.
\citet{EGH:2012} and \citet{GH:2012} propose a notion of multivariate quantile based on Brenier's Theorem. \citet{CCG:2014} (coetaneous with the present paper) propose a conditional version of the optimal transport quantiles of \citet{EGH:2012} and \citet{GH:2012} and apply it to quantile regression, whereas \citet{CGHH:2017} (also coetaneous with this paper) apply optimal transport quantiles to the definition of statistical depth, ranks and signs. However, these papers do not consider identification. The present work is, to the best of our knowledge, the first to apply the notion of multivariate quantiles based on optimal transport results to the identification of simultaneous equations, thus providing a multivariate extension of \citet{Matzkin:2003}'s quantile identification idea.

\subsubsection*{\protect Organization of the paper}

The remainder of the paper is organized as follows. Section~%
\ref{sec:model} sets the hedonic equilibrium framework out. Section~\ref%
{sec:univariate} gives a brief account of the methodology and main results on nonparametric
identification of preferences in single attribute hedonic models, mostly
drawn from \citet{EHN:2004} and \citet{HMN:2010}. Section~\ref%
{sec:multiple} shows how these results
and the shape restrictions that drive them can be extended to the case of
multiple attribute hedonic equilibrium markets.
Section~\ref{sec:se} derives multivariate shape restrictions to identify nonseparable simultaneous equations models. The last section concludes.
Proofs of the main results are
relegated to the appendix, as are necessary background definitions and results on optimal transport
theory and a list of our notational conventions.


\section{Hedonic equilibrium and the identification
problem}

\label{sec:model} We consider a competitive environment,
where consumers and producers trade a good or contract, fully characterized
by its type or quality $z$. The set of feasible qualities $Z\subseteq
\mathbb{R}^{d_z}$ is assumed compact and given a priori, but the distribution of the qualities actually traded arises
endogenously in the hedonic market equilibrium, as does their price schedule
$p(z)$. Producers are characterized by their type $\tilde y\in \tilde
Y\subseteq \mathbb{R}^{d_{\tilde y}}$ and consumers by their type $\tilde
x\in \tilde X\subseteq \mathbb{R}^{d_{\tilde x}}$. Type distributions $%
P_{\tilde x}$ on $\tilde X$ and $P_{\tilde y}$ on $\tilde Y$ are given
exogenously, so that entry and exit are not modelled. Consumers and producers
are price takers and maximize quasi-linear utility $U(\tilde x,z)-p(z)$ and
profit $p(z)-C(\tilde y,z)$ respectively. Utility $U(\tilde x,z)$
(respectively cost $C(\tilde y,z)$) is upper (respectively lower)
semicontinuous and bounded. In addition, the set of qualities $Z(\tilde x,\tilde y)$
that maximize the joint surplus $U(\tilde x,z)-C(\tilde y,z)$ for each pair of
types $(\tilde x,\tilde y)$ is assumed to have a measurable selection. Then, %
\citet{Ekeland:2010} and \citet{CMN:2010} show that an equilibrium
exists in this market, in the form of a price function $p$ on $Z$, a joint
distribution $P_{\tilde xz}$ on $\tilde X\times Z$ and $P_{\tilde yz}$ on $%
\tilde Y\times Z$ such that their marginal on $Z$ coincide, so that market
clears for each traded quality $z\in Z$. Uniqueness is not guaranteed, in
particular prices are not uniquely defined for non traded qualities in
equilibrium. Purity is not guaranteed either: an equilibrium specifies a
conditional distribution $P_{z|\tilde x}$ (respectively $P_{z|\tilde y}$) of
qualities consumed by type $\tilde x$ consumers (respectively produced by
type $\tilde y$ producers). The quality traded by a given producer-consumer
pair $(\tilde x,\tilde y)$ is not uniquely determined at equilibrium without additional assumptions.

\citet{Ekeland:2010} and \citet{CMN:2010} further
show that a pure equilibrium exists and is unique, under the following additional
assumptions. First, type distributions $P_{\tilde x}$ and $P_{\tilde y}$ are
absolutely continuous. Second, gradients of utility and cost, $\nabla_{\tilde
x}U(\tilde x,z)$ and $\nabla_{\tilde y}C(\tilde y,z)$ exist and are
injective as functions of quality $z$. The latter condition, also known as
the {\em Twist Condition} in the optimal transport literature, ensures that all
consumers of a given type $\tilde x$ (respectively all producers of a given type $%
\tilde y$) consume (respectively produce) the same quality $z$
at equilibrium.

The identification problem consists in the recovery of
structural features of preferences and technology from observation of traded
qualities and their prices in a single market.
The solution concept we impose in our identification analysis is the following feature of hedonic equilibrium, i.e., maximization of surplus generated by a trade.
\begin{myassump}{EC}[Equilibrium concept]
\label{ass:ec}
The joint distribution $\gamma$ of $(\tilde X,Z,\tilde Y)$ and the price function $p$ form an hedonic equilibrium, i.e., they satisfy the following.
The joint distribution $\gamma$ has marginals~$P_{\tilde x}$ and~$P_{\tilde y}$ and for~$\gamma$ almost all~$(\tilde x,z,\tilde y)$, \begin{eqnarray*}
U(\tilde x,z)-p(z)&=&\max_{z'\in Z}\left( U(\tilde x,z')-p(z')\right),\\
p(z)-C(\tilde y,z)&=&\max_{z'\in Z}\left( p(z')-C(\tilde y,z)\right).
\end{eqnarray*}
In addition, observed qualities $z\in Z(\tilde x,\tilde y)$ maximizing joint surplus $U(\tilde x,z)-C(\tilde y,z)$ for each $\tilde x\in\tilde X$ and $\tilde y\in\tilde Y$, lie in the interior of the set of feasible qualities $Z$ and $Z(\tilde x,\tilde y)$ is assumed to have a measurable selection. The joint surplus $U(\tilde x,z)-C(\tilde y,z)$ is finite everywhere. We assume full participation in the market\footnote{
The
possibility of non participation can be modelled by adding isolated
points to the sets of types and renormalizing distributions accordingly (see
Section~1.1 of \citet{CMN:2010} for details).
}.
\end{myassump}
Given observability of prices
and the fact that producer type $\tilde y$ (respectively consumer type $%
\tilde x$) does not enter into the utility function $U(\tilde x,z)$
(respectively cost function $C(\tilde y,z)$) directly, we may consider the
consumer and producer problems separately and symmetrically (see \citet{EHN:2002a}). We focus on the
consumer problem and on identification of utility function $U(\tilde x,z)$.
Under assumptions ensuring purity and uniqueness of equilibrium, the model
predicts a deterministic choice of quality $z$ for a given consumer type $%
\tilde x$. We do not impose such assumptions, but we need
to account for heterogeneity in consumption patterns even in case of unique and pure equilibrium. Hence, we assume, as
is customary, that consumer types~$\tilde x$ are only partially observable
to the analyst. We write~$\tilde x=(x,\varepsilon)$, where $x\in X\subseteq\mathbb{R}%
^{d_x}$ is the observable part of the type vector, and~$\varepsilon\in%
\mathbb{R}^{d_\varepsilon}$ is the unobservable part.
We shall make a
separability assumption that will allow us to specify constraints on the
interaction between consumer unobservable type $\varepsilon$ and good
quality~$z$ in order to identify interactions between observable type~$x$
and good quality~$z$.

\begin{myassump}{H}[Unobservable heterogeneity]
\label{ass:uh} Consumer type~$\tilde x$ is composed of
observable type~$x$ with distribution~$P_x$ on~$X\subseteq\mathbb{R}^{d_x}$ and
unobservable type~$\varepsilon$ with a priori specified conditional distribution~$%
P_{\varepsilon\vert x}$ on~$\mathbb{R}^{d_\varepsilon}$, with~$d_\varepsilon\leq d_z$.
The utility of
consumers can be decomposed as~$U(\tilde x,z)=\bar
U(x,z)+\zeta(x,\varepsilon,z)$, where the functional form of~$\zeta$ is known,
but that of~$\bar U$ is~not\footnote{Despite the notation used,
$\bar U$ should not necessarily be interpreted as {\em mean utility},
since we allow for a general choice of $\zeta$ and $P_{\varepsilon\vert x}$.
If this interpretation is desirable in a particular application,~$\zeta$
and~$P_{\varepsilon\vert x}$ can be chosen in such a way that $\mathbb E[\zeta(x,\varepsilon,z)\vert x]=0$.}.
\end{myassump}
The main primitive object of interest is the deterministic component of
utility $\bar U(x,z)$. For convenience, we shall use the transformation~$V(x,z):=p(z)-\bar U(x,z)$. The latter will be called the consumer's {\em potential}, in line with the optimal transport terminology. Since the price is assumed to be identified, identification of~$V$ is equivalent to identification of~$\bar U$.
To achieve identification, we require fixing a choice of function~$\zeta$: a leading example, discussed in Section~\ref{sec:lmu}, is the case~$\zeta(x,\varepsilon,z)=z'\varepsilon$.
Identification also requires fixing the conditional distribution~$P_{\varepsilon\vert x}$ of unobserved heterogeneity. This corresponds to the normalization of the distribution of scalar unobservable utility in existing quantile identification strategies. Discrete choice models also typically rely on a fixed distribution for unobservable heterogeneity (generally extreme valued). The requirements to fix both~$\zeta$ and~$P_{\varepsilon\vert x}$ will be relaxed to some extent in Section~\ref{sec:mm} which entertains the possibility of further identification power using information from multiple markets.


\section{\protect Single market identification with scalar
attribute}

\label{sec:univariate}

In this section, we recall and reformulate
results of \citet{HMN:2010} on identification of single attribute
hedonic models. Suppose, for the purpose of this section, that $%
d_z=d_\varepsilon=1$, so that unobserved heterogeneity is scalar, as is the
quality dimension. Suppose also that $\zeta$ is
twice continuously differentiable in~$z$ and~$\varepsilon$.
\begin{myassump}{R}[Regularity of preferences and technology]
\label{ass:supp} The functions $\bar U(x,z)$, $C(\tilde y,z)$ and $\zeta(x,\varepsilon,z)$ are twice continuously differentiable with respect~$\varepsilon$ and~$z$.
\end{myassump}
Suppose further (for ease of exposition) that~$V$ is
twice continuously differentiable in~$z$.
The main identifying assumption is a shape restriction on utility called single crossing, Spence-Mirlees or supermodularity, depending on the context.
\begin{myassump}{S1}[Spence-Mirlees]
\label{ass:sm} We have $d_z=1$ and $\zeta_{\varepsilon z}(x,\varepsilon,z)>0$ for all $x,\varepsilon,z$.
\end{myassump}
The first order condition of the consumer problem yields
\begin{eqnarray}  \label{eq:foc}
\zeta_z(x,\varepsilon,z)=V_z(x,z),
\end{eqnarray}
which, under Assumption~\ref{ass:sm}, implicitly defines an inverse demand
function $z\mapsto\varepsilon(x,z)$, which specifies which unobserved type
consumes quality $z$. Combining the second order condition $%
\zeta_{zz}(x,\varepsilon,z)<V_{zz}(x,z)$ and further differentiation of (\ref%
{eq:foc}), i.e., $\zeta_{zz}(x,\varepsilon,z)+\zeta_{\varepsilon
z}(x,\varepsilon,z)\varepsilon_z(x,z)=V_{zz}(x,z)$, yields
\begin{eqnarray*}
\varepsilon_z(x,z)=\frac{V_{zz}(x,z)-\zeta_{zz}(x,\varepsilon,z)}{%
\zeta_{\varepsilon z}(x,\varepsilon,z)}>0.
\end{eqnarray*}
Hence the inverse demand is increasing and is therefore identified as the
unique increasing function that maps the distribution $P_{z|x}$ to the
distribution $P_{\varepsilon\vert x}$, namely the quantile transform. Denoting $F$
the cumulative distribution function corresponding to the distribution~$P$,
we therefore have identification of inverse demand according to the strategy
put forward in \citet{Matzkin:2003} as:
\begin{eqnarray*}
\varepsilon(x,z)=F^{-1}_{\varepsilon|x}\left( F_{z|x}(z|x) \right).
\end{eqnarray*}
The single crossing condition of Assumption~\ref{ass:sm} on the consumer
surplus function $\zeta(x,\varepsilon,z)$ yields positive assortative
matching, as in the \citet{Becker:1973} classical model. Consumers with
higher taste for quality $\varepsilon$ will choose higher qualities in
equilibrium and positive assortative matching drives identification of
demand for quality. The important feature of Assumption~\ref{ass:sm} is
injectivity of $\zeta_z(x,\varepsilon,z)$ relative to $\varepsilon$ and a
similar argument would have carried through under $\zeta_{z\varepsilon}(x,\varepsilon,z)<0
$, yielding negative assortative matching instead.

Once inverse demand is identified, the consumer potential $V(x,z)$,
hence the utility function~$\bar U(x,z)$, can be recovered
up to a constant by integration of the first order condition~(\ref{eq:foc}):
\begin{eqnarray*}
\bar U(x,z)=p(z)-V(x,z)=p(z)-\int_0^z\zeta_z(x,\varepsilon(x,z^{\prime}),z^{\prime})dz^{\prime}.
\end{eqnarray*}
We summarize the previous discussion in the following identification statement, originally due to \citet{HMN:2010}.
\begin{proposition}
\label{prop:ehn}
Under Assumptions~\ref{ass:ec}, \ref{ass:uh}, \ref{ass:supp}, \ref{ass:sm}, $\bar U(x,z)$ is nonparametrically identified, in the sense that $z\mapsto \bar U_z(x,z)$ is the only marginal utility function compatible with the pair~$(P_{xz},p)$,
i.e., any other marginal utility function coincides with it, $P_{z\vert x}$ almost surely.
\end{proposition}
Unlike the demand function, which is identified without knowledge of the
surplus function~$\zeta$, as long as the latter satisfies single crossing (Assumption~\ref%
{ass:sm}), identification of the preference function~$\bar U(x,z)$ does
require a priori knowledge of the function~$\zeta$.


\section{\protect Single market identification with multiple
attributes}

\label{sec:multiple}

This section develops the multivariate analogue of identification results in Section~\ref{sec:univariate}.
The strategy follows the same lines. First, a shape restriction on the utility function, analogue to Assumption~\ref{ass:sm}, and a consequence of
maximization behavior, called {\em cyclical monotonicity}, will identify the inverse demand: we will show that a single type~$\varepsilon(x,z)$
chooses good quality~$z$ at equilibrium. Second, formally, the utility is then recovered from the first order condition of the consumer's program.
The latter step, however, involves significant difficulties, due to the possible lack of differentiability of the endogenous price function.

\subsection{Identification of inverse demand}

\subsubsection{Shape restriction}

In the one dimensional case, identification of inverse demand was shown under the single crossing Assumption~\ref{ass:sm}. We noted
that the sign of the single crossing condition was not important for the
identification result. Instead, what is crucial is the following, weaker, condition, which is
commonly known as the {\em Twist Condition} in the optimal transport literature. The crucial condition, maintained throughout, is Assumption~\ref{ass:twist}. As shown in \citet{CMN:2010}, it holds when for each distinct pair of consumer types~$(\varepsilon_1,\varepsilon_2)$, the function~$z\mapsto\zeta(x,\varepsilon_1,z)-\zeta(x,\varepsilon_2,z)$ has no critical point. It is satisfied for example in the case~$\zeta(x,\varepsilon,z)=z'\varepsilon$, in the case~$\zeta(x,\varepsilon,z)=F(x,z'\varepsilon),$ when~$z'\varepsilon>0,$ and $F$ is increasing and convex in its second argument, or in the case~$\zeta(x,\varepsilon,z)=\sum_{k=1}^{d}F_k(x,\varepsilon_k,z_k),$ where, for each~$k$ and~$x$, $F_k$ is supermodular in~$(\varepsilon_k,z_k)$. All these examples are discussed in Sections~\ref{sec:lmu} and~\ref{sec:exp} below.
\begin{myassump}{S2}[Twist Condition]
\label{ass:twist}
For all $x$ and $z$, the following hold.
\begin{enumerate}
\item[(A)]
The gradient $\nabla_z\zeta(x,\varepsilon,z)$ of $\zeta(x,\varepsilon,z)$ in
$z$ is injective as a function of $\varepsilon\in\mbox{Supp}(P_{\varepsilon\vert x})$.
\item[(B)] The gradient $\nabla_\varepsilon\zeta(x, \varepsilon,z)$ of $\zeta(x,\varepsilon,z)$ in
$\varepsilon$ is injective as a function of $z\in Z$.
\end{enumerate}
\end{myassump}
From \citet{GN:65}, it is sufficient that $D^2_{z\varepsilon}\zeta(x,\varepsilon,z)$ be positive definite everywhere for Assumption~\ref{ass:twist} to be satisfied.
Alternative sets of sufficient conditions are given in Theorem~2 of \citet{Mas-Colell:79}\footnote{Relatedly, \citet{BGH:2013} propose injectivity results under a gross substitutes condition.}.
Assumption~\ref{ass:twist}, unlike single crossing, is well defined in the multivariate case, and we shall show,
using recent developments in optimal transport theory, that it continues to
deliver the desired identification in the multivariate case.

\subsubsection{Cyclical monotonicity}
An important implication of Assumption~\ref{ass:ec} is that traded quality $z$
maximizes the joint surplus\footnote{This observation is the
basis for the characterization of hedonic models as transferable utility matching models in \citet{CMN:2010}.}~$U(\tilde x,z)-C(\tilde y,z)$.
Let, therefore,~$S(x,\varepsilon,\tilde y):=\sup_{z\in Z}[U((x,\varepsilon),z)-C(\tilde y,z)]$ be the surplus of a consumer-producer pair~$((x,\varepsilon),\tilde y)$ at equilibrium.
Suppose consumer~$(x,\varepsilon_0)$ (resp.~$(x,\varepsilon_1)$) is paired at equilibrium with producer~$\tilde y_0$ (resp.~$\tilde y_1$) to exchange good quality~$z$.
Then, as shown in the proof of Lemma~\ref{lemma:idid} in the Appendix, their total surplus is at least as large in the current consumer-producer matching as
it would be, were they to switch partners. This is a property of the optimal allocation called {\em cyclical monotonicity}. The total surplus cannot be improved by a cycle of reallocations of consumers and producers. Applied here to cycles of length two, cyclical monotonicity yields:
\begin{eqnarray*}
S(x,\varepsilon_0,\tilde y_0) + S(x,\varepsilon_1,\tilde y_1) &\geq& S(x,\varepsilon_0,\tilde y_1) + S(x,\varepsilon_1,\tilde y_0)\\
&\geq& U((x, \varepsilon_0), z) -C(\tilde y_1,z)+U((x, \varepsilon_1), z) -C(\tilde y_0, z) \\
&=&S(x,\varepsilon_0,\tilde y_0) + S(x,\varepsilon_1,\tilde y_1),
\end{eqnarray*}
where the first inequality holds because of cyclical monotonicity and the second inequality holds by definition of the surplus function~$S$
as a supremum. Hence, equality holds throughout in the previous display. Therefore, choice~$z$ maximizes both $z\mapsto U((x, \varepsilon_0), z) -C(\tilde y_0,z)$
and $z\mapsto U((x, \varepsilon_1), z) -C(\tilde y_0, z)$. It follows that~$\nabla_z\zeta(x,\varepsilon_1,z)=\nabla_z\zeta(x,\varepsilon_0,z)$. The Twist Assumption~\ref{ass:twist} then yields
equality of~$\varepsilon_0$ and~$\varepsilon_1$ and the following lemma (proved formally in the Appendix).

\begin{lemma}[Identification of inverse demand]
\label{lemma:idid}
Under Assumptions~\ref{ass:ec}, \ref{ass:uh}, \ref{ass:twist}(A) and~\ref{ass:supp},
if consumers with characteristics~$(x,\varepsilon_0)$ and~$(x,\varepsilon_1)$ consume the same good quality~$z$
at equilibrium, then~$\varepsilon_0=\varepsilon_1$.
\end{lemma}
We see from Lemma~\ref{lemma:idid} that identification of inverse demand holds under conditions that are analogous to
the scalar case, where the Twist condition replaces Spence-Mirlees as a shape restriction. We also see that the identification proof
also relies on a notion of monotonicity. We will push this analogy further in the Section~\ref{sec:se} and show that the inverse demand function~$z\mapsto\varepsilon(x,z)$
itself satisfies a generalized form of monotonicity, we call~$\zeta$-monotonicity, since its definition involves the function~$\zeta$.

\subsection{Identification of marginal utility}

Heuristically, once identification of inverse demand is established, and~$\varepsilon(x,z)$ is uniquely defined in equilibrium,
the first order condition of the consumer's problem~$\sup_{z\in Z}\{\zeta(x,\varepsilon,z)-V(x,z)\}$
delivers identification of marginal utility~$\nabla_z\bar U(x,z)$. However, using the first order condition presupposes
smoothness of the potential~$V$, hence of the endogenous price function~$z\mapsto p(z)$. Conditions for differentiability of the potential~$V$
in optimal transportation problems are given by \citet{MTW:2005}. They are applied to identification in hedonic models in \citet{Nesheim:2013}.
However, the conditions in \citet{MTW:2005} are not transparent and they are known to be very strong, excluding simple forms of the surplus, such as~$S(\tilde x,\tilde y):=\vert \tilde x-\tilde y\vert^p$, for all~$p\ne2$ for instance. The remainder of this section, therefore, will be devoted to proving a weaker form of differentiability
of the price function, which can be used to identify marginal utility from the first order condition of the consumer's problem.

The consumer's problem yields the expression for indirect utility:
\begin{eqnarray}  \label{eq:zconv}
\sup_{z\in Z}\{\zeta(x,\varepsilon,z)-V(x,z)\}:=V^\zeta(x,\varepsilon).
\end{eqnarray}
Equation~(\ref{eq:zconv}) defines a generalized notion of convex
conjugation, in which the consumer's indirect utility is the conjugate of~$V(x,z)=p(z)-\bar U(x,z)$. This notion of conjugation can be inverted, similarly to
convex conjugation, into:
\begin{eqnarray}  \label{eq:max}
V^{\zeta\zeta}(x,z)=\sup_{\varepsilon\in\mathbb R^{d_\varepsilon}}\{\zeta(x,\varepsilon,z)-V^%
\zeta(x,\varepsilon)\},
\end{eqnarray}
where $V^{\zeta\zeta}$ is called the {\em double conjugate} of $V$.
In the special case~$\zeta(x,\varepsilon,z)=z'\varepsilon$,~$\zeta$-conjugate simplifies to the ordinary
convex conjugate of convex analysis and convexity of a lower semi-continuous function is equivalent to equality with its double conjugate. Hence, by extension,
equality with its double~$\zeta$-conjugate defines a generalized notion of convexity (see Definition~2.3.3 page~86 of \citet{Villani:2003}).
\begin{definition}[$\protect\zeta$-convexity]
\label{def:zconv} A function $V$ is called $\zeta$-convex if $V=V^{\zeta\zeta}$.
\end{definition}
We establish~$\zeta$-convexity of the potential as a step towards a notion of differentiability, in analogy with
convex functions, which are locally Lipschitz, hence almost surely differentiable, by Rademacher's Theorem (see for instance \citet{Villani:2009}, Theorem~10.8(ii)).
It also delivers a notion of~$\zeta$-monotonicity for inverse demand, discussed in Section~\ref{sec:se}.
\begin{lemma}
\label{lemma:zconv}
Under~\ref{ass:ec} and~\ref{ass:uh}, the function~$z\mapsto V(x,z)$ is~$P_{z\vert x}$ a.s.~$\zeta$-convex, for all $x$.
\end{lemma}
This result only provides information for pairs~$(x,z)$, where type~$x$ consumers choose good quality~$z$ in equilibrium. In order to obtain a global smoothness result on the potential~$V$, we need conditions under which the endogenous distribution of good qualities traded in equilibrium is absolutely continuous with respect to Lebesgue measure on~$\mathbb R^{d_z}$. They include absolute continuity of the distribution of unobserved heterogeneity, additional smoothness conditions on preferences and technology and the Twist Assumption~\ref{ass:twist}(B), which requires the dimension of unobserved heterogeneity to be the same as the dimension of the good quality space, i.e.,~$d_\varepsilon=d_z$ (this will be relaxed in Section~\ref{sec:ld}).

\begin{myassump}{H'}
\label{ass:uh'} Assumption~\ref{ass:uh} holds and the distribution of unobserved tastes $P_{\varepsilon\vert x}$ is absolutely continuous on~$\mathbb R^{d_z}$ with respect to Lebesgue measure for all $x$.
\end{myassump}

\begin{myassump}{R'}[Conditions for absolute continuity of $P_{z\vert x}$]
\label{ass:supp3} The following hold.
\begin{enumerate}
\item The Hessian of total surplus
$D^2_{zz} \left( U(x,\varepsilon,z) - C(\tilde y,z) \right) $ is bounded above; that is $\Vert D^2_{zz} \left( U(x,\varepsilon,z) - C(\tilde y,z) \right) \Vert \leq M_1$ for all $x,\varepsilon,z,\tilde y$, for some fixed $M_1$.
\item If the support of $P_{\varepsilon\vert x}$ is unbounded, for each $x\in X$, $\Vert\nabla_z\zeta(x,\varepsilon,z)\Vert\rightarrow\infty$ as $\|\varepsilon\|\rightarrow\infty$, uniformly in $z\in Z$.
\item The matrix $D^2_{\varepsilon z }\zeta(x,\varepsilon,z)$ has full rank for all $x, \varepsilon, z$. Its inverse $[D^2_{\varepsilon z }\zeta(x,\varepsilon,z)]^{-1}$ has uniform upper bound~$M_0$:
$\left\Vert[D^2_{\varepsilon z }\zeta(x,\varepsilon,z)]^{-1}\right\Vert\leq M_0$ for all $x,\varepsilon,z$, for some fixed~$M_0$.
\end{enumerate}
\end{myassump}

We then obtain our principal intermediate lemma, which is of independent interest for the theory of hedonic equilibrium. The proof can be found in the technical appendix \citet{CGHP:2020app}.
\begin{lemma}
\label{lemma:abs3}
Under Assumptions~\ref{ass:ec}, ~\ref{ass:uh'},~\ref{ass:twist} and~\ref{ass:supp3}, the endogenous distribution $P_{z\vert x}$ of qualities traded at equilibrium is absolutely continuous with respect to Lebesgue measure.
\end{lemma}
Lemma~4 in the technical appendix \citet{CGHP:2020app} shows everywhere differentiability of the double conjugate potential~$V^{\zeta\zeta}$. This doesn't imply differentiability of~$V$ everywhere, since~$V$ is~$\zeta$-convex (i.e.,~$V=V^{\zeta\zeta}$) only~$P_{z\vert x}$ almost everywhere. However, combined with Lemma~\ref{lemma:abs3}, this yields approximate differentiability of~$z\mapsto V(x,z)$ as defined in Definition~6 of the technical appendix \citet{CGHP:2020app}. Using uniqueness of the approximate gradient of an approximately differentiable function then yields identification of marginal utility~$\nabla_z\bar U(x,z)$ from the first order condition
\begin{eqnarray}
\label{eq:FOC}
\nabla_z\zeta(x,\varepsilon(x,z),z)=\nabla_{ap,z}\;V(x,z)=\nabla_{ap} \; p(z)-\nabla_z\bar U(x,z),
\end{eqnarray}
where~$\nabla_{ap,(z)}$ denotes the approximate gradient (with respect to~$z$) of Definition~6 of the technical appendix \citet{CGHP:2020app}, and where the inverse demand function $\varepsilon(x,z)$ is uniquely determined, by Lemma~\ref{lemma:idid}. This yields our main identification theorem.

\begin{theorem}
\label{thm:id2} Under Assumptions~\ref{ass:ec}, ~\ref{ass:uh'},~\ref{ass:twist} and~\ref{ass:supp3}, the following hold:
\begin{enumerate}
\item $\bar U(x,z)$ is
nonparametrically identified, in the sense that $z\mapsto \nabla_z\bar U(x,z)$ is the only marginal utility function compatible with the
pair $(P_{xz},p)$,
i.e., any other marginal utility function coincides with it, $P_{z\vert x}$ almost surely.
\item For all $x\in X$, $\bar U(x,z)=p(z)-V(x,z)$ and $z\mapsto V(x,z)$ is $P_{z\vert x}$ almost everywhere equal to the $\zeta$-convex solution to the problem $\min_V\left(\mathbb{E}_z[V(x,z)\vert x]+\mathbb{E}_\varepsilon%
[V^\zeta(x,\varepsilon)\vert x]\right)$, with $V^\zeta$ defined in~(\ref{eq:zconv}).
\end{enumerate}
\end{theorem}
Theorem~\ref{thm:id2}(1) provides identification of marginal utility without any restriction on the distributions of observable characteristics of producers and consumers. The latter may include discrete characteristics. Regularity conditions in Assumption~\ref{ass:supp3} are satisfied in the cases~$\zeta(x,\varepsilon,z)=\varepsilon'z$ and~$\zeta(x,\varepsilon,z)=\exp(\varepsilon'z)$ as we discuss in Sections~\ref{sec:lmu} and~\ref{sec:exp} respectively. They preclude bunching of consumers at equilibrium, as shown in Lemma~\ref{lemma:abs3}. The result is driven by the shape restriction~\ref{ass:twist} and the strong normalization assumption on the distribution of unobserved heterogeneity. This assumption is inevitable in a single market identification based on a generalized quantile identification strategy. Section~\ref{sec:mm} discusses (partial) identification without knowledge a priori of the distribution of unobserved heterogeneity. Theorem~\ref{thm:id2}(2) provides a framework for estimation and inference on the identified marginal utility based on new developments in computational optimal transport (see for instance the survey in \citet{PC:2019}).

\subsection{Special cases}

\subsubsection{Marginal utility linear in unobserved taste}
\label{sec:lmu}

A leading special case of the identification result in Theorem~\ref{thm:id2} is the choice $\zeta(x,\varepsilon,z)=z^{\prime}\varepsilon$, where marginal utility is linear in unobservable taste.
A natural interpretation of this specification is that each quality dimension~$z_j$ is associated with a specific unobserved taste intensity $\varepsilon_j$ for that particular quality dimension.
Assumptions~\ref{ass:twist} and~\ref{ass:supp3}(2,3) are automatically satisfied when~$\zeta(x,\varepsilon,z)=\varepsilon'z$. In addition, $\zeta$-convexity reduces to traditional convexity, so that we have the following corollary
of Theorem~\ref{thm:id2}:

\begin{corollary}
\label{cor:id2}
Under~\ref{ass:ec}, ~\ref{ass:uh'} with~$\zeta(x,\varepsilon,z)=\varepsilon'z$, and~\ref{ass:supp3}(1), the following hold:
\begin{enumerate}
\item $\bar U(x,z)$ is
nonparametrically identified, in the sense that $z\mapsto \nabla_z\bar U(x,z)$ is the only marginal utility function compatible with the
pair $(P_{xz},p)$,
i.e., any other marginal utility function coincides with it, $P_{z\vert x}$ almost surely.
\item For all $x\in X$, $\bar U(x,z)=p(z)-V(x,z)$ and $z\mapsto V(x,z)$ is $P_{z\vert x}$ almost everywhere equal to the convex solution to the problem $\min_V\left(\mathbb{E}_z[V(x,z)\vert x]+\mathbb{E}_\varepsilon%
[V^\ast(x,\varepsilon)\vert x]\right)$, where $V^\ast$ is the convex conjugate of~$V$.
\end{enumerate}
\end{corollary}
A significant computational advantage of Corollary~\ref{cor:id2}(2) over the general case is that the potential solves a convex program. The first order condition~(\ref{eq:FOC}) also simplifies to
$\varepsilon(x,z)=\nabla_{ap,z}\;V(x,z)$, where~$\nabla_{ap,z}$ denotes the approximate gradient with respect to~$z$ (Definition~6 of the technical appendix \citet{CGHP:2020app}). Hence, inverse demand in this case is the approximate gradient of a~$P_{z\vert x}$ almost surely convex function. This can be interpreted as a multivariate version of monotonicity of inverse demand and assortative matching, since in the univariate case, we recover monotonicity of inverse demand as in Section~\ref{sec:univariate}.

\subsubsection{Nonlinear transforms and random Barten scales}
\label{sec:exp}

In the special case of Section~\ref{sec:lmu}, the curvature of marginal willingness to pay for qualities varies only with observable characteristics, but is independent of unobservable type~$\varepsilon$. Two ways Specification~$\zeta(z,\varepsilon,z)=z'\varepsilon$ can be generalized to allow the curvature of marginal utility to vary with unobserved type are the following.
\begin{enumerate}

\item Specification~$\zeta(x,\varepsilon,z)=F(x,z'\varepsilon)$ satisfies Assumptions~\ref{ass:twist} and~\ref{ass:supp3} when the conditional distribution~$P_{\varepsilon\vert x}$ of types has bounded support,~$z'\varepsilon>0,$ and~$F$ is increasing and convex in its second argument (with non zero derivative). A special case is~$\zeta(x,\varepsilon,z):=\exp(z'\varepsilon)$. Under this specification, a type~$\varepsilon$ consumer's marginal willingness to pay for quality~$z$ is~$\nabla_z\bar U(x,z)+\exp(z'\varepsilon)\varepsilon$. The curvature of marginal willingness to pay increases with type, so that there are increasing returns to quality.

\item Specification~$\zeta(x,\varepsilon,z)=\sum_{k=1}^{d_z}F_k(x,\varepsilon_k,z_k),$ satisfies Assumption~\ref{ass:twist} when~$F_k$ is supermodular in~$(\varepsilon_k,z_k)$  for each~$k$ and~$x$, and satifies Assumption~\ref{ass:supp3} when~$\partial^2F_k/\partial z_k^2$ is bounded above and~$\partial^2F_k/\partial\varepsilon_k\partial z_k$ is bounded below for each~$k$, uniformely over~$x$. A special case is~$\zeta(x,\varepsilon,z)=\sum_{k=1}^{d}F_k(x,z_k\varepsilon_k),$ in the spirit of random Barten scales, as in \citet{LP:2017}.

\end{enumerate}

\subsection{Extensions}

\subsubsection{Lower dimensional unobserved heterogeneity}
\label{sec:ld}

Our main identification result, Theorem~\ref{thm:id2}, is obtained under conditions that force the dimension of unobserved heterogeneity
to be the same as the dimension~$d_z$ of the good quality space. The conditions that impose~$d_\varepsilon=d_z$ are Assumption~\ref{ass:uh'} which requires the distribution~$P_{\varepsilon\vert x}$ to be absolutely continuous with respect to Lebesgue measure on~$\mathbb R^{d_z}$, and Assumption~\ref{ass:twist}(B), which requires injectivity of~$z\mapsto\nabla_\varepsilon\zeta(x,\varepsilon,z)$. In the special case~$\zeta(x,\varepsilon,z)=\varepsilon'z$, the interpretation of each dimension of~$\varepsilon$ as a quality dimension specific taste is appealing. However, the choice of dimension of unobserved heterogeneity remains an arbirary modelling choice.

In this section, we relax these assumptions and analyze identification with unobserved heterogeneity of lower dimension, including~$d_\varepsilon=0$ and~$d_\varepsilon=1$. First recall that inverse demand is identified in Lemma~\ref{lemma:idid} under Assumptions~\ref{ass:ec}, \ref{ass:uh}, \ref{ass:twist}(A) and~\ref{ass:supp}, which only require~$d_\varepsilon\leq d_z$. We also know from Lemma~\ref{lemma:zconv} that the potential~$z\mapsto V(x,z)$ is~$P_{z\vert x}$ almost everywhere~$\zeta$-convex. It is therefore~$\zeta$-convex on any open subset of the support of~$P_{z\vert x}$, which we show implies local identification. To obtain a global identification result, we need this~$\zeta$-convexity everywhere.

\begin{myassump}{S3}[$\protect\zeta$-convexity]
\label{ass:zconv} The potential $V$ is $\zeta$-convex as a function of
$z$ for all~$x$.
\end{myassump}
Unfortunately, this global constraint implies a constraint on the endogenous price function, for which we do not have sufficient conditions in the general case.
Under Assumption~\ref{ass:zconv}, we show differentiability of the potential function~$z\mapsto V(x,z)$, hence identification of marginal utility.

\begin{theorem}[Lower dimensional unobservable heterogeneity]
\label{thm:id}
\;\mbox{ }\;
\begin{enumerate}
\item Local identification: Under Assumptions~\ref{ass:ec}, ~\ref{ass:uh},~\ref{ass:supp},~\ref{ass:supp3}(2)
and~\ref{ass:twist}(A),
\begin{enumerate}
\item $\nabla_z\bar U(x,z)$ is identified on any open subset of the support of~$P_{z\vert x}$,
\item $v'\nabla_z\bar U(x,z)$ is identified for any vector~$v$ tangent to the support of~$P_{z\vert x}$.
\end{enumerate}
\item Global identification: If, in addition, Assumption~\ref{ass:zconv} holds, $\nabla_z\bar U(x,z)$ is identified.
\end{enumerate}
\end{theorem}
Local identification therefore holds under very weak assumptions, as seen in Theorem~\ref{thm:id}(1a). However, in cases with lower dimensional unobserved heterogeneity, consumer choices may be concentrated on a lower dimensional manifold, so that there are no open subsets in the support of~$P_{z\vert x}$. In such cases, Theorem~\ref{thm:id}(1b) tells us that we can only identify marginal willingness to pay along the support of good attributes actually traded at equilibrium. To illustrate the idea, suppose consumers are acquiring housing. The latter is differentiated along two dimensions, size and air quality, say. Suppose consumers with identical observable characteristics~x are heterogeneous along a single scalar dimension of unobserved heterogeneity~$\varepsilon$. We would then expect equilibrium housing choices to be concentrated on a curve in the (size~$\times$~air quality) space, implicitly defining a scalar quality index that is monotonic in unobserved type~$\varepsilon$. Our result says that we can identify counterfactual marginal willingness to pay for size and air quality along that curve of observed equilibrium choices only.
To obtain global identification with lower dimensional unobserved heterogeneity, we need a global~$\zeta$-convexity assumption on the potential (which implies a shape restriction on the endogenous price function).

We investigate special cases:
\begin{itemize}

\item The case~$\zeta(x,\varepsilon,z)=0$. From the first order condition of the consumer's program, we then have~$\nabla p(z)=\nabla_z U(x,z)$, so that~$U(x,z)$ is additively separable in~$x$ and~$z$ and the data on matching between consumers and producers cannot inform utility.

\item Scalar unobserved heterogeneity: suppose~$d_\varepsilon=1$ and in addition~$\zeta(x,\varepsilon,z)=\zeta(\varepsilon,z)$. Denote the (scalar) inverse demand~$z\mapsto\varepsilon(x,z)$ and assume regularity of all the terms involved. Differentiating the first order condition of the consumer problem with respect to consumer observable characteristic~$x$ yields
\[D^2_{xz}\bar U(x,z)=-D^2_{x\varepsilon}\zeta(\varepsilon(x,z),z)\nabla_x\varepsilon(x,z),\] which is at most of rank~$1$ since~$D^2_{x\varepsilon}\zeta(\varepsilon(x,z),z)$ is a~$d_z\times d_\varepsilon$ matrix.

\end{itemize}

\subsubsection{Partial identification with multiple markets}
\label{sec:mm}

All identification results so far, require fixing the distribution of unobserved consumer heterogeneity a priori. In this section, we derive identifying information from multiple markets, and the possibility of jointly (partially) identifying the utility function~$\bar U(x,z)$ and the distribution~$P_{\varepsilon\vert x}$.
Suppose~ $m_1$ and $m_2$ index two separate markets, in the sense that producers, consumers or goods cannot move between markets. Markets differ in the distributions of producer and consumer characteristics~$(P_x^{m_1},P_{\tilde y}^{m_1})$ and $(P_x^{m_2},P_{\tilde y}^{m_2})$. Suppose, however, that the distribution of unobserved tastes~$P_{\varepsilon\vert x}$ and the utility function~$U(x,\varepsilon,z)=\bar U(x,z)+z'\varepsilon$ is identical in both markets.
Both markets are at equilibrium. The equilibrium price schedule in market~$m$ is~$p^{m}(z)$. The equilibrium distribution of traded qualities in market~$m$ is~$P^m_{z\vert x}$.

Under the assumptions of Corollary~\ref{cor:id2}, in each market, we recover a nonparametrically identified utility function $\bar U^{m}(x,z;P_{\varepsilon\vert x})$, where the dependence in the unknown distribution of tastes $P_{\varepsilon\vert x}$ is emphasized. For each fixed~$P_{\varepsilon\vert x}$, Corollary~\ref{cor:id2} tells us that~$\nabla_z\bar U^{m}(x,z;P_{\varepsilon\vert x})$ is uniquely determined. In each market, the first order condition of the consumer's problem is~$\varepsilon(x,z;m)=\nabla_{ap,z}p^m(z)-\nabla_z\bar U(x,z;P_{\varepsilon\vert x})$. Differencing across markets therefore yields:
\begin{eqnarray}\label{eq:mm}
\varepsilon(z,x;m_1)-\varepsilon(x,z;m_2)=\nabla_{ap,z}\left( p^{m_1}(z)-p^{m_2}(z)\right),
\end{eqnarray}
which is an identifying equation for~$P_{\varepsilon\vert x}$. The right-hand side of~(\ref{eq:mm}) is identified, since the price functions are observed. Moreover, for each market~$m$,
the inverse demand~$\varepsilon(x,z;m)$ uniquely determines~$P_{\varepsilon\vert x}$, since it pushes the identified~$P^m_{z\vert x}$ forward to~$P_{\varepsilon\vert x}$. Point identification would require conditions under which the difference~$\varepsilon(z,x;m_1)-\varepsilon(x,z;m_2)$ uniquely determines~$P_{\varepsilon\vert x}$, which is beyond the scope of the present work.

\section{Identification of non separable simultaneous equations}
\label{sec:se}

The analysis of hedonic equilibrium models in Section~\ref{sec:multiple} motivates a new approach to the identification of nonseparable
simultaneous equations models of the type~$H(x,z)=\varepsilon$, where~$x\in\mathbb R^{d_x}$ is an observed vector of covariates, $z\in\mathbb R^{d_z}$ is the vector of dependent variables,~$P_{z\vert x}$ is identified from the data,~$H$ is an unknown function and~$\varepsilon\in\mathbb R^{d_\varepsilon}$ is a vector of unobservable shocks with distribution~$P_{\varepsilon\vert x}$. In the case~$d_\varepsilon=d_z=1$, $H$ is identified by \citet{Matzkin:2003} subject to the normalization of~$P_{\varepsilon\vert x}$ and monotonicity of~$z\mapsto H(x,z)$ for all~$x$. This section develops a class of shape restrictions that allows identification of~$H$ in the multivariate case~$1\leq d_\varepsilon\leq d_z$.

As in the scalar case, we fix the conditional distribution~$P_{\varepsilon|x}$ of errors a priori. This is justified by the fact that for any vector of dependent variables~$Z\sim P_{z\vert x}$ and any pair of absolutely continuous error distributions~$(P_{\varepsilon\vert x},\tilde P_{\tilde\varepsilon\vert x})$, there is an invertible mapping~$T$ such that~$H(x,Z)\sim P_{\varepsilon\vert x}$ and~$\tilde H(x,Z):=H(x,T(Z))\sim \tilde P_{\tilde\varepsilon\vert x}$ (by \citet{McCann:95}), so that~$(H,P_{\varepsilon\vert x})$ and~$(\tilde H,\tilde P_{\tilde\varepsilon\vert x})$ are observationally equivalent.

For identification, we rely on a shape restriction that emulates monotonicity in~$z$ of~$H(x,z)$ in the scalar case.
This generalized monotonicity notion is inherited from utility maximizing choices of good quality~$z$ by consumers with characteristics~$(x,\varepsilon)$.
As such, it is indexed by the utility function.
\begin{definition}[$\zeta$-monotonicity]
\label{def:zmon}
Let~$\zeta$ be a function on~$\mathbb R^{d_x}\times\mathbb R^{d_\varepsilon}\times\mathbb R^{d_z}$ that is continuously differentiable in its second and third variables and satisfies Assumption~\ref{ass:twist}(A). A function~$H$ on~$\mathbb R^{d_x}\times\mathbb R^{d_z}$ is~$\zeta$-monotone if for all~$x$, there exists a~$\zeta$-convex function~$z\mapsto V(x,z)$ (see Definition~\ref{def:zconv}) such that for all~$z$, $H(x,z)\in\partial^{\zeta}_zV(x,z)$, where~$\partial^{\zeta}_z$ denotes the~$\zeta$-subdifferential with respect to~$z$ from Definition~4 of the technical appendix \citet{CGHP:2020app}.
\end{definition}
Two special cases help clarify the concept of~$\zeta$-monotonicity:
\begin{enumerate}
\item When~$d_z=d_\varepsilon=1$, injectivity of~$\varepsilon\mapsto\zeta_z(x,\varepsilon,z)$ implies that~$V^\zeta$ is convex or concave, and that~$z\mapsto H(x,z)$ is monotone.
\item When~$d_\varepsilon=d_z$ and~$\zeta(x,\varepsilon,z)=z'\varepsilon$, $V^\zeta=V^\ast$, which is convex, and therefore locally Lipschitz, hence almost surely differentiable by Rademacher's Theorem (\citet{Villani:2009}, Theorem~10.8(ii)), so that~$H(x,z)=\nabla_zV(x,z)$ is the gradient of a convex function.
\end{enumerate}

The class of~$\zeta$-monotone functions has structural underpinnings as demand functions resulting from the maximization of a utility function over good qualities~$z\in\mathbb R^{d_z}$. Suppose a consumer with characteristics~$(x,\varepsilon)$ chooses~$z$ based on the maximization of~$\bar U(x,z)+\zeta(x,\varepsilon,z)-p(z)$. Suppose~$\zeta$ satisfies Assumption~\ref{ass:twist} and~$V(x,z):=p(z)-\bar U(x,z)$ is~$\zeta$-convex. Then, the demand function~$H$ that satisfies~$\nabla V(x,z)=\nabla_z\zeta(x,H(x,z),z)$, $P_{z\vert x}$ a.s., is~$\zeta$-monotonic by Theorem~10.28(b) page~243 of \citet{Villani:2009}.
Theorem~\ref{thm:matzkin} below then shows identification of demand when utility is of the form~$\bar U(x,z)+\zeta(x,\varepsilon,x)$ and~$\zeta$ is fixed.
\begin{theorem}[Nonseparable simultaneous equations]
\label{thm:matzkin}
In the simultaneous equations model $H(x,z)=\varepsilon$, with $z\in \mathbb R^{d_z}$ and $x\in \mathbb R^{d_x}$, and~$\varepsilon$ follows known distribution~$P_{\varepsilon\vert x}$,
the function $H:\mathbb R^{d_x}\times\mathbb R^{d_z}\rightarrow\mathbb R^{d_\varepsilon}$ is identified within the class of measurable~$\zeta$-monotone
functions, for any function~$\zeta$ on~$\mathbb R^{d_x}\times\mathbb R^{d_\varepsilon}\times\mathbb R^{d_z}$ that is bounded above, continuously differentiable in its second and third variables and satisfies Assumptions~\ref{ass:twist}(A), \ref{ass:supp3}(2) and $\int \zeta(x,\varepsilon,z)dP_{\varepsilon\vert x}(\varepsilon) \geq C\in\mathbb R$ for all~$x,z$.
\end{theorem}
Theorem~\ref{thm:matzkin} is a relatively straightforward application of classical results in optimal transport theory, in particular Theorem~10.28 page~243 of \citet{Villani:2009}. Brenier's polar factorization theorem, in \citet{Brenier:91}, was, to the best of our knowledge, first used to define multivariate quantile functions by \citet{EGH:2012} and \citet{GH:2012} with decision theoretic applications. \citet{CCG:2014} and \citet{CGHH:2017}  (both coetaneous with the present paper) apply \citet{McCann:95} to multivariate quantile regression and multivariate depth, quantiles, ranks and signs respectively. This section relies on an extension of these optimal transport results to more general transport costs and interprets it as an identification result, thus extending scalar quantile identification strategies.

If we revisit the two special cases of~$\zeta$-monotonicity above, we obtain the classical quantile identification of \citet{Matzkin:2003}, and a result on the identification of nonseparable simultaneous equations systems within the class of gradients of convex functions.
\begin{enumerate}
\item When~$d_z=d_\varepsilon=1$, Theorem~\ref{thm:matzkin} yields the identification of~$H$ in the system~$\varepsilon=H(x,z)$ when~$z\mapsto H(x,z)$ is monotone.
\item When~$d_\varepsilon=d_z$ and~$\zeta(x,\varepsilon,z)=z'\varepsilon$, Theorem~\ref{thm:matzkin} yields the identification of~$H$ in the system~$\varepsilon=H(x,z)$ when~$z\mapsto H(x,z)$ is the gradient of a convex function.
\end{enumerate}

\begin{corollary}
\label{cor:McCann}
In the simultaneous equations model $z=G(x,\varepsilon)$, with $z\in \mathbb R^{d_z}$, $x\in \mathbb R^{d_x}$, $\varepsilon\sim P_{\varepsilon\vert x}$, and~$P_{\varepsilon\vert x}$ a given absolutely continuous distribution on~$\mathbb R^{d_z}$,
the function $G:\mathbb R^{d_x}\times\mathbb R^{d_z}\rightarrow\mathbb R^{d_z}$ is identified within the class of gradients of convex functions of~$\varepsilon$, for each~$x$.
\end{corollary}
Although the previous result is presented as a corollary of Theorem~\ref{thm:matzkin}, it holds under weaker conditions than would be implied by Theorem~\ref{thm:matzkin} in case~$\zeta(x,\varepsilon,z)=z'\varepsilon$ and is a direct application of the {\em Main Theorem}~ in \citet{McCann:95}. The only constraint is the absolute continuity of the distribution of~$\varepsilon$, so that the outcome vector~$z$ and the covariate vector~$x$ are unrestricted.

Beyond the special cases of Corollary~\ref{cor:McCann}, we revisit the examples of Section~\ref{sec:exp}. First consider the case of exponential transform~$\zeta(x,\varepsilon,z):=\exp(z'\varepsilon)$. Given marginal utility~$\bar U(x,z)+\exp(z'\varepsilon)$ and price~$p(z)$, Theorem~\ref{thm:matzkin} tells us that the solution~$\varepsilon=H(x,z)$ to the consumer's problem~$\nabla p(z)-\nabla_z\bar U(x,z)=\exp(z'H(x,z))H(x,z)$ is unique. In the case consumers are maximizing utility of the form~$\bar U(x,z)+\sum_{k=1}^dF_k(x,z_k\varepsilon_k)$ , the corresponding system of differential equations with a unique solution is~$p_{z_k}(z)-\bar U_{z_k}(x,z)=f_k(x,z_kH_k(x,z))H_k(x,z)$, each~$k$, where~$p_{z_k}$ and~$\bar U_{z_k}$ are the partial derivatives with respect to the~$k$-th variable,~$f_k$ is the derivative of~$F_k$ with respect to the second argument, and~$H_k$ is the~$k$-th component of~$H$.

\section*{Discussion}

This paper proposed a set of conditions under which utilities and costs in a hedonic equilibrium model are identified from the observation of a single market outcome.
The proof strategy extends \citet{EHN:2004} and \citet{HMN:2010} (hereafter EHMN) to the case of goods characterized by more than one attribute. The proposed shape restriction on the utility function, called twist condition, extends the single crossing condition in EHMN. The proof of identification mirrors that of (one of the strategies in) EHMN. First, inverse demand is identified from the twist condition and cyclical monotonicity (a feature of equilibrium). Then the first order condition of the consumer's problem allows the recovery of the utility function, once a suitable form of weak differentiability of the endogenous price function is ensured. The identification proof highlights another parallel with EHMN, which is (generalized) monotonicity of inverse demand. In the scalar case, this generalized monotonicity reduces to monotonicity, whereas in the special case, where utility takes the form~$U((x,\varepsilon),z)=\bar U(x,z)+z'\varepsilon$, inverse demand is the gradient of a convex function. We then show that this generalized form of monotonicity is a suitable shape restriction to identify nonseparable simultaneous equations models with a strategy that extends the quantile identification of \citet{Matzkin:2003}.
Most of our results involve fixing the distribution of unobserved consumer heterogeneity a priori, as in the original quantile identification method. Although we provide some discussion of the case, where data from multiple distinct markets can provide additional identifying equations to (partially) identifying $(\bar U(x,z),P_{\varepsilon\vert x})$ jointly, more research is needed to develop point identification conditions for the latter.


\appendix

{\footnotesize

\section{Notation}

Throughout the paper, we use the following notational conventions. Let $f(x,y)$ be a real-valued function on $\mathbb R^d\times \mathbb R^d$. When $f$ is sufficiently smooth, the gradient of~$f$ with respect to~$x$ is denoted~$\nabla_xf$, the matrix of second order derivatives with respect to~$x$ and~$y$ is denoted~$D^2_{xy}f$. When~$f$ is not smooth,~$\partial_xf$ refers to the subdifferential with respect to~$x$, from Definition~2 of the technical appendix \citet{CGHP:2020app}, and~$\nabla_{ap,x}f$ refers to the approximate gradient with respect to~$x$, from Definition~6 of the technical appendix \citet{CGHP:2020app}. The set of all Borel probability distributions on a set~$Z$ is denoted~$\Delta(Z)$. A random vector~$\varepsilon$ with probability distribution~$P$ is denoted~$\varepsilon\sim P$, and~$X\sim Y$ means that the random vectors~$X$ and~$Y$ have the same distribution. The product of two probability distributions~$\mu$ and~$\nu$ is denoted~$\mu\otimes\nu$ and for a map~$f:X\mapsto Y$ and~$\mu\in\Delta(X)$,~$\nu:=f\#\mu$ is the probability distribution on~$Y$ defined for each Borel subset~$A$ of~$Y$ by~$\nu(A)=\mu( f^{-1}(A))$. For instance, if~$T$ is a map from~$X$ to~$Y$ and~$\nu$ a probability distribution on~$X$, then~$\mu:=(id,T)\#\nu$ defines the probability distribution on~$X\times Y$ by~$\mu(A)=\int_X1_A(x,T(x))d\mu(x)$ for any measurable subset~$A$ of~$X\times Y$. Given two probability distributions~$\mu$ and~$\nu$ on~$X$ and~$Y$ respectively,~$\mathcal M(\mu,\nu)$ will denote the subset of~$\Delta(X\times Y)$ containing all probability distributions with marginals~$\mu$ and~$\nu$. We denote the inner product of two vectors~$x$ and~$y$ by~$x^\prime y$. The Euclidean norm is denoted~$\|\cdot\|$. The notation~$|a|$ refers to the absolute value of the real number~$a$, whereas~$|A|$ refers to the Lebesgue measure of set~$A$. The set of all continuous real valued functions on~$Z$ is denoted~$C^0(Z)$ and~$B_r(x)$ is the open ball of radius~$r$ centered at~$x$. For each fixed $x\in X$, the function $\varepsilon\mapsto V^\ast(x,\varepsilon):=\sup_{z\in\mathbb R^{d}}\{z'\varepsilon-V(x,\varepsilon)\}$ is called the {\em convex conjugate} (also known as {\em Legendre-Fenchel transform})
of $z\mapsto V(x,z)$. Still for fixed $x$, the convex conjugate $z\mapsto V^{\ast\ast}(x,\varepsilon):=\sup_{\varepsilon\in\mathbb R^{d}}\{z'\varepsilon-V^\ast(x,\varepsilon)\}$ is called the {\em double conjugate} or {\em convex envelope} of the potential function $z\mapsto V(x,z)$. According to convex duality theory (see for instance \citet{Rockafellar:70}, Theorem~12.2 page 104), the double conjugate of $V$ is $V$ itself if and only if $z\mapsto V(x,z)$ is convex and lower-semi-continuous. The notation~$\psi^\zeta$ refers to the~$\zeta$-convex conjugate of the function~$\psi$ (Definition~3 of the technical appendix \citet{CGHP:2020app}) and~$\partial^\zeta_x$ refers to the~$\zeta$-subdifferential with respect to~$x$ (Definition~4 of the technical appendix \citet{CGHP:2020app}).

}

{\footnotesize

\section{Proof of results in the main text}\label{app:proofs}

\subsubsection*{{\protect\footnotesize Proof of Lemma~\protect\ref%
{lemma:zconv}}}

{\footnotesize
By definition of $V^\zeta$, we have \begin{equation}\label{dualinequality}
V(x,z) \geq \zeta(x,\varepsilon,z) -V^{\zeta}(x,\varepsilon).
\end{equation}
As, by definition of $\zeta$-conjugation,
$V^{\zeta\zeta}(x,z) =\sup_{\varepsilon} \left[ \zeta(x,\varepsilon,z) -V^{\zeta}(x,\varepsilon) \right],$
we have
\begin{equation}
\label{firstinequality}
V(x,z) \geq V^{\zeta\zeta}(x,z),
\end{equation}
by taking supremum over $\varepsilon$ in \eqref{dualinequality}.

Let $\gamma$ be an hedonic equilibrium probability distribution on $\tilde X\times Z\times\tilde Y$. By Assumption~\ref{ass:ec},
\[
\zeta(x,\varepsilon,z)-V(x,z)=U(x,\varepsilon,z)-p(z)=\max_{z\in Z}\left(U(x,\varepsilon,z)-p(z)\right)=\max_{z\in Z}\left(\zeta(x,\varepsilon,z)-V(x,z)\right)=V^\zeta(x,\varepsilon)
\]
is true~$\gamma$ almost everywhere. Hence, there is equality in (\ref{dualinequality}) $\gamma$ almost everywhere.
Hence, for $P_{z|x}$ almost every $z$, and $\varepsilon$ such that $(z, \varepsilon)$ is in the support of $\gamma$, we have~$V(x,z) = \zeta(x,\varepsilon,z) -V^{\zeta}(x,\varepsilon).$
But the right hand side is bounded above by $V^{\zeta\zeta}(x,z)$ by definition, so we get
$V(x,z ) \leq V^{\zeta\zeta}(x,z).$
Combined with \eqref{firstinequality}, this tells us $V(x,z) =V^{\zeta\zeta}(x,z)$, $P_{z|x}$ almost everywhere.
\hfill%
$\square$ }

\subsubsection*{Proof of Lemma~\ref{lemma:idid}}

\footnotesize{

For a fixed observable type $x$, assume that the types $\tilde x_0:=(x,\varepsilon_0)$ and $\tilde x_1:=(x,\varepsilon_1)$ both choose the same good, $\bar z \in Z$, from producers $\tilde y_0$ and $\tilde y_1$, respectively.

  We want to prove that this implies the unobservable types are also the same; that is, that $\varepsilon_0=\varepsilon_1$.  This property is equivalent to having a map from the good qualities $Z$ to the unobservable types, for each fixed observable type.

Note that $\bar z$ must maximize the joint surplus for both $\varepsilon_0$ and $\varepsilon_1$.  That is, setting
\begin{equation*}\label{maxsurplus}
S(x,\varepsilon,\tilde y) = \sup_{z\in Z}[\bar U(x, z) +\zeta(x, \varepsilon, z) -C(\tilde y, z)]
\end{equation*}
we have,
\begin{equation}\label{surplusmax_0}
S(x,\varepsilon_0,\tilde y_0)= \bar U(x,\bar z) +\zeta(x, \varepsilon_0,\bar  z) -C(\tilde y_0,\bar z)
\end{equation}
 and
\begin{equation}\label{surplusmax_1}
S(x,\varepsilon_1,\tilde y_1)= \bar U(x,\bar z) +\zeta(x, \varepsilon_1,\bar  z) -C(\tilde y_1,\bar z).
\end{equation}

By Assumption~\ref{ass:ec}, we can apply Lemma~1 of \citet{CMN:2010}, so that the pair of indirect utilities $(V,W)$, where
$V(\tilde x)=\sup_{z\in Z}\left( U(\tilde x,z)-p(z)\right)$ and~$W(\tilde y)=\sup_{z\in Z}\left( p(z)-C(\tilde y,z)\right)$,
achieve the dual (DK) of the optimal transportation problem
\begin{eqnarray*}
\sup_{\pi\in\mathcal M\left(P_{\tilde x},P_{\tilde y}\right)}\int S(\tilde x,\tilde y)d\pi(\tilde x,\tilde y),
\end{eqnarray*}
with solution $\pi$.
This implies, from Theorem~1.3 page 19 of \citet{Villani:2003}, that for $\pi$ almost all pairs $(\tilde x_0,\tilde y_0)$ and $(\tilde x_1,\tilde y_1)$,
\begin{eqnarray*}
V(\tilde x_0)+W(\tilde y_0)&=&S(\tilde x_0,\tilde y_0),\\
V(\tilde x_1)+W(\tilde y_1)&=&S(\tilde x_1,\tilde y_1),\\
V(\tilde x_0)+W(\tilde y_1)&\geq&S(\tilde x_0,\tilde y_1),\\
V(\tilde x_1)+W(\tilde y_0)&\geq&S(\tilde x_1,\tilde y_0).
\end{eqnarray*}
  We therefore deduce the condition (called the $2$-monotonicity condition):
$$
S(x,\varepsilon_0,\tilde y_0)+S(x,\varepsilon_1,\tilde y_1) \geq S(x,\varepsilon_1,\tilde y_0)+S(x,\varepsilon_0,\tilde y_1),
$$  recalling that $\tilde x_0=(x,\varepsilon_0)$ and $\tilde x_1=(x,\varepsilon_1)$.
 Now, by definition of $S$ as the maximized surplus, we have
\begin{equation}\label{surplusineq0}
S(x,\varepsilon_1,\tilde y_0) \geq \bar U(x, \bar z) +\zeta(x, \varepsilon_1, \bar z) -C(\tilde y_0, \bar z)
\end{equation}
 and
\begin{equation}\label{surplusineq1}
S(x,\varepsilon_0,\tilde y_1) \geq \bar U(x,\bar  z) +\zeta(x, \varepsilon_0, \bar z) -C(\tilde y_1,\bar z).
\end{equation}
 Inserting this, as well as \eqref{surplusmax_0} and \eqref{surplusmax_1} into the $2$-monotonicity inequality yields
\begin{eqnarray*}
 \bar U(x,\bar z) +\zeta(x, \varepsilon_0,\bar  z) -C(\tilde y_0,\bar z)&+&
\bar U(x,\bar z) +\zeta(x, \varepsilon_1,\bar  z) -C(\tilde y_1,\bar z)\\
 &\geq&  S(x,\varepsilon_1,\tilde y_0)+S(x,\varepsilon_0,\tilde y_1) \\
&\geq&  \bar U(x, \bar z) +\zeta(x, \varepsilon_1, \bar z) -C(\tilde y_0,\bar z)+\bar U(x,\bar z) +\zeta(x, \varepsilon_0, \bar z) -C(\tilde y_1, \bar z).
\end{eqnarray*}
But the left and right hand sides of the preceding string of inequalities are identical, so we must have equality throughout.  In particular, we must have equality in \eqref{surplusineq0} and \eqref{surplusineq1}.  Equality in  \eqref{surplusineq0}, for example, means that $\bar z$ maximizes $z \mapsto \bar U(x, z) +\zeta(x, \varepsilon_1, z) -C(\tilde y_0, z)$, and so, as $\bar z$ is in the interior of $Z$ by Assumption~\ref{ass:ec},  we have
\begin{equation}\label{zderivative1}
 \nabla _z\zeta(x, \varepsilon_1, \bar z) =\nabla _zC(\tilde y_0, \bar z)-\nabla _z\bar U(x, \bar z) .
\end{equation}
Since $\bar z$ also maximizes $z \mapsto \bar U(x, z) +\zeta(x, \varepsilon_0, z) -C(\tilde y_0, z)$, we also have
\begin{equation}\label{zderivative0}
 \nabla _z\zeta(x, \varepsilon_0, \bar z) =\nabla _zC(\tilde y_0, \bar z)-\nabla _z\bar U(x, \bar z).
\end{equation}
Equations \eqref{zderivative1} and \eqref{zderivative0} then imply
$
\nabla _z\zeta(x, \varepsilon_1, \bar z) =\nabla _z\zeta(x, \varepsilon_0, \bar z)
$
and Assumption~\ref{ass:twist}(A) implies $ \varepsilon_1 = \varepsilon_0$.

\hfill%
$\square$
}

\subsubsection*{{\protect\footnotesize Proof of Theorem~\ref{thm:id2}}}

{\footnotesize

(1) Since by Lemma~4 in the technical appendix \citet{CGHP:2020app}, $V(x,z)$ is approximately differentiable~$P_{z\vert x}$ almost surely, and since $\bar U(x,z)$ is differentiable by assumption, $p(z)=V(x,z)+\bar U(x,z)$ is also approximately differentiable~$P_{z\vert x}$ almost surely. Since, by Lemma~\ref{lemma:idid}, the inverse demand function $\varepsilon(x,z)$ is uniquely determined,
the first order condition $\nabla_z\zeta(x,\varepsilon(x,z),z)=\nabla_{ap} p(z)-\nabla_z\bar U(x,z)$ identifies
$\nabla_z\bar U(x,z)$, $P_{z\vert x}$ almost everywhere, as required. }

(2) In Part (1), we have shown uniqueness (up to location) of the pair $(V,V^\zeta)$ such that~$V(x,z)+V^\zeta(x,\varepsilon)=\zeta(x,\varepsilon,z)$, $\pi$ almost surely. By Theorem~1.3 page 19 of \citet{Villani:2003}
, this implies that~$(V,V^\zeta)$ is the unique (up to location) pair of $\zeta$-conjugates that solves the dual Kantorovitch problem as required.
\hfill%
$\square$
}

\subsubsection*{\footnotesize Proof of Theorem~\ref{thm:id}(1a) and Theorem~\ref{thm:id}(2)}

{\footnotesize

\underline{Step 1}: Differentiability of $V$ in $z$.
The objective is to prove that the subdifferential (see Definition~2 of the technical appendix \citet{CGHP:2020app}) at each $z_0$ is a singleton, which is equivalent to differentiability at~$z_0$.
We show Theorem~\ref{thm:id}(2). The same method of proof applies on any open subset of the support of~$P_{z\vert x}$ to yield the proof of Theorem~\ref{thm:id}(1a). From Assumption~\ref{ass:zconv}, $V(x,z)$ is $\zeta$-convex, and hence locally semiconvex (see Definition~5 of the technical appendix \citet{CGHP:2020app}), by Proposition C.2 in \citet{GM:96}. We recall the definition of local semiconvexity from the latter paper.
Now, Lemma~\ref{lemma:idid} shows that for each fixed $z$, the set
\begin{equation}\label{equalityset}
\{\varepsilon\in\mathbb R^{d_z}: V(x,z) +V^{\zeta}(x,\varepsilon) = \zeta(x,\varepsilon,z)\}:=\{f(z)\}
\end{equation}
is a singleton.  We claim that this means $V$ is differentiable with respect to $z$ everywhere.  Fix a point~$z_0$. We will prove that the subdifferential $\partial_z V(x,z_0)$ contains only one extremal point (for a definition, see \citet{Rockafellar:70}, Section~18). This will yield the desired result. Indeed, all subdifferentials are closed and convex. Hence, so is the subdifferential of $V$. By Assumption~\ref{ass:zconv},~$V$ is~$\zeta$-convex, hence continuous, by the combination of Propositions~C.2 and~C.6(i) in \citet{GM:96}. Hence, as the subdifferential of a continuous function, the subdifferential of~$V$ is also bounded. Hence, it is equal to the convex hull of its extreme points (see \citet{Rockafellar:70}, Theorem~18.5). The subdifferential of $V$ at $z_0$ must therefore be a singleton, and $V$ must be differentiable at $z_0$ (Theorem~25.1 in \citet{Rockafellar:70} can be easily extended to locally semiconvex functions).
Let $q$ be any extremal point in $\partial_z V(x,z_0)$.  Let $z_n$ be a sequence satisfying the conclusion in Lemma~3 in the technical appendix \citet{CGHP:2020app}.  Now, as~$V$ is differentiable at each point $z_n$, we have the envelope condition
\begin{equation}\label{envelope}
\nabla_z V(x,z_n) = \nabla_z \zeta (x,\varepsilon_n,z_n)
\end{equation}
where $\varepsilon_n=f(z_n)$ is the unique point giving equality in \eqref{equalityset}.

As the sequence $\nabla_z\zeta(x,\varepsilon_n,z_n)$ converges, Assumption~\ref{ass:supp3}(2) implies that
the $\varepsilon_n$ remain in a bounded set.
We can therefore pass to a convergent subsequence $\varepsilon_n \rightarrow \varepsilon_0$.
By continuity of $\nabla_z\zeta$, we can pass to the limit in  \eqref{envelope} and, recalling that the left hand side tends to $q$, we obtain
$q=\nabla_z\zeta (x,\varepsilon_0,z_0).$
Now, by definition of $\varepsilon_n$, we have the equality $V(x,z_n) +V^{\zeta}(x,\varepsilon_n) = \zeta(x,\varepsilon_n,z_n).$
By Assumption~\ref{ass:zconv}, $V$ and $V^\zeta$ are $\zeta$-convex, hence continuous, by the combination of Propositions~C.2 and~C.6(i) in \citet{GM:96}. Hence, we can pass to the limit to obtain
$V(x,z_0) +V^{\zeta}(x,\varepsilon_0) = \zeta(x,\varepsilon_0,z_0).$
But this means~$\varepsilon_0 =f(z_0)$, and so $q=\nabla_z\zeta (x,\varepsilon_0,z_0) = \nabla_z\zeta (x,f(z_0),z_0) $ is uniquely determined by $z_0$.  This means that the subdifferential can only have one extremal point, completing the proof of differentiability of~$V$.

\underline{Step 2}: Since by Step 1, $V(x,z)$ is differentiable~$P_{z\vert x}$ almost surely, and since $\bar U(x,z)$ is differentiable by assumption, $p(z)=V(x,z)+\bar U(x,z)$ is also differentiable~$P_{z\vert x}$ almost surely. Since, by Lemma~\ref{lemma:idid}, the inverse demand function $\varepsilon(x,z)$ is uniquely determined,
the first order condition $\nabla_z\zeta(x,\varepsilon(x,z),z)=\nabla p(z)-\nabla_z\bar U(x,z)$ identifies
$\nabla_z\bar U(x,z)$, $P_{z\vert x}$ almost everywhere, as required%
.\hfill$\square$ }

\subsubsection*{{\protect\footnotesize Proof of Theorem~\protect\ref{thm:id}(1b)}}
{\footnotesize
The argument in the proof of Theorem~\ref{thm:id}(2) applies to~$V^{\zeta\zeta}(x,z)$; this function is therefore differentiable throughout the support of~$P_{z|x}$.
Now, if~$v$ is tangent to the support of~$P_{z|x}$ at $z_0$, there is a curve~$z(t)$ in~$P_{z|x}$  which is differentiable at~$t_0$, where~$t_0$ is such that~$z_0=z(t_0)$, and~$z'(t_0)=v$.  Since $V^{\zeta\zeta}(x,z) =V(x,z)$ on the support, we have $V^{\zeta\zeta}(x,z(t)) =V(x,z(t))$.  Since the left hand side is differentiable as a function of~$t$ at~$t_0$, the right hand side must be as well, and:
$$
\nabla_zV^{\zeta\zeta}(x,z(t))\cdot v=\nabla_zV^{\zeta\zeta}(x,z(t))\cdot z'(t_0) =\frac{\partial }{\partial t}[V(x,z(t))]|_{t=t_0}.
$$
Therefore, $p(z) =V(x,z)+\bar U(x,z)$ is also differentiable along this curve, and  $$\frac{\partial }{\partial t}[p(z(t))]|_{t=t_0} =\frac{\partial }{\partial t}[V(x,z(t))]|_{t=t_0}+\nabla_z\bar U(x,z(t_0))\cdot v.$$  It then follows from the first order condition that
$$
\nabla_{z}\zeta(x,\epsilon(x,z),z)\cdot v =\nabla_zV^{\zeta\zeta}(x,z(t))\cdot v= \frac{\partial }{\partial t}[V(x,z(t))]|_{t=t_0}=\frac{\partial }{\partial t}[p(z(t))]|_{t=t_0}-\nabla_z\bar U(x,z(t_0))\cdot v,
$$
which identifies the direction derivative $\nabla_z\bar U(x,z(t_0))\cdot v$.
\hfill%
$\square$
}

\subsubsection*{{\protect\footnotesize Proof of Theorem~\protect\ref%
{thm:matzkin}}}

{\footnotesize
Fix~$x\in\mathbb R^{d_x}$ and omit from the notation throughout the proof. Fix a twice continuously differentiable function~$\zeta$ on~$\mathbb R^{d_\varepsilon}\times\mathbb R^{d_z}$ that satisfies Assumption~\ref{ass:twist}(A). Assume there exist two~$\zeta$-monotonic functions~$H$ and~$\tilde H$ such that $H\# P_z=P_\varepsilon$. By the definition of~$\zeta$-monotonicity (Definition~\ref{def:zmon}), there exist two~$\zeta$-convex functions~$\psi$ and~$\tilde\psi$ such that~$H\in\partial^\zeta\psi$ and~$\tilde H\in\partial^\zeta\tilde\psi$. By definition of
the~$\zeta$-subdifferential (Definition~4 of the technical appendix \citet{CGHP:2020app}), it implies that~$P_z$ almost surely,
$\psi(z)+\psi^\zeta(H(z))=\zeta(H(z),z)$ and~$\tilde\psi(z)+\tilde\psi^\zeta(H(z))=\zeta(\tilde H(z),z)$.
Hence, both~$H$ and~$\tilde H$ are soutions to the Monge optimal transport problem with cost~$\zeta$, which is unique by Theorem~10.28 page 243 of \citet{Villani:2009}. It remains to show that the assumptions of Theorem~10.28 page 243 of \citet{Villani:2009} are satisfied. Indeed, by Assumption~\ref{ass:supp}, $\zeta$ is differentiable everywhere, hence superdifferentiable, verifying~(i); (ii) is Assumption~\ref{ass:twist} and~(iii) is satisfied under Assumption~\ref{ass:supp3}(2) by Step 1 of the proof of Theorem~\ref{thm:id}(2). Finally, integrating $
  \int \zeta(x,\epsilon,z)dP_{\epsilon|x}(\epsilon)> C$ over~$P_{z\vert x}$ with~$\zeta$ bounded above yields~$\int\!\!\!\int \zeta(x,\epsilon,z)dP_{z|x}(z)dP_{\epsilon|x}(\epsilon)$ finite, which is the remaining condition for Theorem~10.28 page 243 of \citet{Villani:2009} to hold.  The result follows.
\hfill$\square$ }

\section{Online Supplement}

\subsection{Hedonic equilibrium}

\label{sec:model} We consider a competitive environment,
where consumers and producers trade a good or contract, fully characterized
by its type or quality $z$. The set of feasible qualities $Z\subseteq
\mathbb{R}^{d_z}$ is assumed compact and given a priori, but the distribution of the qualities actually traded arises
endogenously in the hedonic market equilibrium, as does their price schedule
$p(z)$. Producers are characterized by their type $\tilde y\in \tilde
Y\subseteq \mathbb{R}^{d_{\tilde y}}$ and consumers by their type $\tilde
x\in \tilde X\subseteq \mathbb{R}^{d_{\tilde x}}$. Type distributions $%
P_{\tilde x}$ on $\tilde X$ and $P_{\tilde y}$ on $\tilde Y$ are given
exogenously, so that entry and exit are not modelled. Consumers and producers
are price takers and maximize quasi-linear utility $U(\tilde x,z)-p(z)$ and
profit $p(z)-C(\tilde y,z)$ respectively. Utility $U(\tilde x,z)$
(respectively cost $C(\tilde y,z)$) is upper (respectively lower)
semicontinuous and bounded
. In addition, the set of qualities $Z(\tilde x,\tilde y)$
that maximize the joint surplus $U(\tilde x,z)-C(\tilde y,z)$ for each pair of
types $(\tilde x,\tilde y)$ is assumed to have a measurable selection. Then, %
\citet{Ekeland:2010} and \citet{CMN:2010} show that an equilibrium
exists in this market, in the form of a price function $p$ on $Z$, a joint
distribution $P_{\tilde xz}$ on $\tilde X\times Z$ and $P_{\tilde yz}$ on $%
\tilde Y\times Z$ such that their marginal on $Z$ coincide, so that market
clears for each traded quality $z\in Z$. Uniqueness is not guaranteed, in
particular prices are not uniquely defined for non traded qualities in
equilibrium. Purity is not guaranteed either: an equilibrium specifies a
conditional distribution $P_{z|\tilde x}$ (respectively $P_{z|\tilde y}$) of
qualities consumed by type $\tilde x$ consumers (respectively produced by
type $\tilde y$ producers). The quality traded by a given producer-consumer
pair $(\tilde x,\tilde y)$ is not uniquely determined at equilibrium without additional assumptions.

The solution concept we impose is the following feature of hedonic equilibrium, i.e., maximization of surplus generated by a trade.
\begin{myassump}{EC}[Equilibrium concept]
\label{ass:ec}
The joint distribution $\gamma$ of $(\tilde X,Z,\tilde Y)$ and the price function $p$ form an hedonic equilibrium, i.e., they satisfy the following.
The joint distribution $\gamma$ has marginals~$P_{\tilde x}$ and~$P_{\tilde y}$ and for~$\gamma$ almost all~$(\tilde x,z,\tilde y)$, \begin{eqnarray*}
U(\tilde x,z)-p(z)&=&\max_{z'\in Z}\left( U(\tilde x,z')-p(z')\right),\\
p(z)-C(\tilde y,z)&=&\max_{z'\in Z}\left( p(z')-C(\tilde y,z)\right).
\end{eqnarray*}
In addition, observed qualities $z\in Z(\tilde x,\tilde y)$ maximizing joint surplus $U(\tilde x,z)-C(\tilde y,z)$ for each $\tilde x\in\tilde X$ and $\tilde y\in\tilde Y$, lie in the interior of the set of feasible qualities $Z$ and $Z(\tilde x,\tilde y)$ is assumed to have a measurable selection. The joint surplus $U(\tilde x,z)-C(\tilde y,z)$ is finite everywhere. We assume full participation in the market.
\end{myassump}
We assume that consumer types~$\tilde x$ are only partially observable
to the analyst. We write~$\tilde x=(x,\varepsilon)$, where $x\in X\subseteq\mathbb{R}%
^{d_x}$ is the observable part of the type vector, and~$\varepsilon\in%
\mathbb{R}^{d_\varepsilon}$ is the unobservable part.
We shall make a
separability assumption that will allow us to specify constraints on the
interaction between consumer unobservable type $\varepsilon$ and good
quality~$z$ in order to identify interactions between observable type~$x$
and good quality~$z$.

\begin{myassump}{H}[Unobservable heterogeneity]
\label{ass:uh} Consumer type~$\tilde x$ is composed of
observable type~$x$ with distribution~$P_x$ on~$X\subseteq\mathbb{R}^{d_x}$ and
unobservable type~$\varepsilon$ with a priori specified conditional distribution~$%
P_{\varepsilon\vert x}$ on~$\mathbb{R}^{d_\varepsilon}$, with~$d_\varepsilon\leq d_z$.
The utility of
consumers can be decomposed as~$U(\tilde x,z)=\bar
U(x,z)+\zeta(x,\varepsilon,z)$, where the functional form of~$\zeta$ is known,
but that of~$\bar U$ is~not.
\end{myassump}
For convenience, we shall use the transformation~$V(x,z):=p(z)-\bar U(x,z)$. The latter will be called the consumer's {\em potential}, in line with the optimal transport terminology.
The following condition is
commonly known as the {\em Twist Condition} in the optimal transport literature.
\begin{myassump}{S2}[Twist Condition]
\label{ass:twist}
For all $x$ and $z$, the following hold.
\begin{enumerate}
\item[(A)]
The gradient $\nabla_z\zeta(x,\varepsilon,z)$ of $\zeta(x,\varepsilon,z)$ in
$z$ is injective as a function of $\varepsilon\in\mbox{Supp}(P_{\varepsilon\vert x})$.
\item[(B)] The gradient $\nabla_\varepsilon\zeta(x, \varepsilon,z)$ of $\zeta(x,\varepsilon,z)$ in
$\varepsilon$ is injective as a function of $z\in Z$.
\end{enumerate}
\end{myassump}
From \citet{GN:65}, it is sufficient that $D^2_{z\varepsilon}\zeta(x,\varepsilon,z)$ be positive definite everywhere for Assumption~\ref{ass:twist} to be satisfied.

\subsection{Absolute continuity of the distribution of traded qualities}

The consumer's problem can be written
\begin{eqnarray}  \label{eq:zconv}
\sup_{z\in Z}\{\zeta(x,\varepsilon,z)-V(x,z)\}:=V^\zeta(x,\varepsilon).
\end{eqnarray}
Equation~(\ref{eq:zconv}) defines a generalized notion of convex
conjugation, which can be inverted, similarly to
convex conjugation, into:
\begin{eqnarray}  \label{eq:max}
V^{\zeta\zeta}(x,z)=\sup_{\varepsilon\in\mathbb R^{d_\varepsilon}}\{\zeta(x,\varepsilon,z)-V^%
\zeta(x,\varepsilon)\},
\end{eqnarray}
where $V^{\zeta\zeta}$ is called the {\em double conjugate} of $V$.
In the special case~$\zeta(x,\varepsilon,z)=z'\varepsilon$,~$\zeta$-conjugate simplifies to the ordinary
convex conjugate of convex analysis and convexity of a lower semi-continuous function is equivalent to equality with its double conjugate. Hence, by extension,
equality with its double~$\zeta$-conjugate defines a generalized notion of convexity (see Definition~2.3.3 page~86 of \citet{Villani:2003}).
\begin{definition}[$\protect\zeta$-convexity]
\label{def:zconv} A function $V$ is called $\zeta$-convex if $V=V^{\zeta\zeta}$.
\end{definition}
We establish~$\zeta$-convexity of the potential as a step towards a notion of differentiability, in analogy with
convex functions, which are locally Lipschitz, hence almost surely differentiable, by Rademacher's Theorem (see for instance \citet{Villani:2009}, Theorem~10.8(ii)).
\begin{lemma}
\label{lemma:zconv}
Under~\ref{ass:ec} and~\ref{ass:uh}, the function~$z\mapsto V(x,z)$ is~$P_{z\vert x}$ a.s.~$\zeta$-convex, for all $x$.
\end{lemma}
This result only provides information for pairs~$(x,z)$, where type~$x$ consumers choose good quality~$z$ in equilibrium. In order to obtain a global smoothness result on the potential~$V$, we need conditions under which the endogenous distribution of good qualities traded in equilibrium is absolutely continuous with respect to Lebesgue measure on~$\mathbb R^{d_z}$. They include absolute continuity of the distribution of unobserved heterogeneity, additional smoothness conditions on preferences and technology and the Twist Assumption~\ref{ass:twist}(B), which requires the dimension of unobserved heterogeneity to be the same as the dimension of the good quality space, i.e.,~$d_\varepsilon=d_z$.

\begin{myassump}{H'}
\label{ass:uh'} Assumption~\ref{ass:uh} holds and the distribution of unobserved tastes $P_{\varepsilon\vert x}$ is absolutely continuous on~$\mathbb R^{d_z}$ with respect to Lebesgue measure for all $x$.
\end{myassump}

\begin{myassump}{R'}[Conditions for absolute continuity of $P_{z\vert x}$]
\label{ass:supp3} The following hold.
\begin{enumerate}
\item The Hessian of total surplus
$D^2_{zz} \left( U(x,\varepsilon,z) - C(\tilde y,z) \right) $ is bounded above; that is $\Vert D^2_{zz} \left( U(x,\varepsilon,z) - C(\tilde y,z) \right) \Vert \leq M_1$ for all $x,\varepsilon,z,\tilde y$, for some fixed $M_1$.
\item For each $x\in X$, $\Vert\nabla_z\zeta(x,\varepsilon,z)\Vert\rightarrow\infty$ as $\|\varepsilon\|\rightarrow\infty$, uniformly in $z\in Z$.
\item The matrix $D^2_{\varepsilon z }\zeta(x,\varepsilon,z)$ has full rank for all $x, \varepsilon, z$. Its inverse $[D^2_{\varepsilon z }\zeta(x,\varepsilon,z)]^{-1}$ has uniform upper bound~$M_0$:
$\left\Vert[D^2_{\varepsilon z }\zeta(x,\varepsilon,z)]^{-1}\right\Vert\leq M_0$ for all $x,\varepsilon,z$, for some fixed~$M_0$.
\end{enumerate}
\end{myassump}

We then obtain absolute continuity of the endogenous distribution of traded qualities.
\begin{lemma}
\label{lemma:abs3}
Under Assumptions~\ref{ass:ec}, ~\ref{ass:uh'},~\ref{ass:twist} and~\ref{ass:supp3}, the endogenous distribution $P_{z\vert x}$ of qualities traded at equilibrium is absolutely continuous with respect to Lebesgue measure.
\end{lemma}
Lemma~\ref{lemma:ad} below shows everywhere differentiability of the double conjugate potential~$V^{\zeta\zeta}$. This doesn't imply differentiability of~$V$ everywhere, since~$V$ is~$\zeta$-convex (i.e.,~$V=V^{\zeta\zeta}$) only~$P_{z\vert x}$ almost everywhere. However, combined with Lemma~\ref{lemma:abs3}, this yields approximate differentiability of~$z\mapsto V(x,z)$ as defined in Definition~\ref{def:apd}.

\subsubsection*{Proof of Lemma~\ref{lemma:abs3}}

{\footnotesize

For an upper semi-continuous map $S:\mathbb R^{d_1}\times\mathbb R^{d_2}\rightarrow\mathbb R$, possibly indexed by $x$, and probability distributions $\mu$ on $\mathbb R^{d_1}$ and $\nu$ on $\mathbb R^{d_2}$, possibly conditional on $x$, let $T_S(\mu,\nu)$ be the solution of the Kantorovich problem, i.e.,
$$
T_S(\mu,\nu) =\sup_{\gamma\in\mathcal M(\mu,\nu) }\int_{\mathbb{R}^{d_1} \times \mathbb{R}^{d_2} } S(\varepsilon, z)d\gamma(\varepsilon, z).
$$
 We state and prove three intermediate results in Steps~1, 2 and~3 before completing the proof of Lemma~\ref{lemma:abs3}.

\underline{Step 1:} We first show that $P_z$ achieves
\begin{eqnarray}
\label{equilibrium1}
\max_{\nu\in\Delta(Z)} [T_U(P_{\tilde x},\nu) +T_{-C}(P_{\tilde y},\nu)].
\end{eqnarray}
This step follows the method of proof of Proposition~3 in \citet{CE:2010}. Take any probability distribution~$\nu\in\Delta(Z)$,~$\mu_1\in\mathcal M(P_{\tilde x},\nu)$ and~$\mu_2\in\mathcal M(P_{\tilde y},\nu)$. By the Disintegration Theorem (see for instance \citet{Pollard:2002}, Theorem~9, page~117), there are families of conditional probabilities $\mu_1^z$ and $\mu_2^z$, $z\in Z$, such that $\mu_1=\mu_1^z\otimes\nu$ and $\mu_2=\mu_2^z\otimes\nu$. Define the probability $\gamma\in\Delta(\tilde X\times\tilde Y\times Z)$ by
\[
\int_{\tilde X\times\tilde Y\times Z} F(\tilde x,\tilde y,z)d\gamma(\tilde x,\tilde y,z)=\int_{\tilde X\times\tilde Y\times Z} F(\tilde x,\tilde y,z)d\mu^z_1(\tilde x)d\mu^z_2(\tilde y)d\nu(z),
\] for each $F\in C^0(\tilde X,\tilde Y,Z)$.
By construction, the projection $\mu$ of $\gamma$ on $\tilde X\times\tilde Y$ belongs to $\mathcal M(P_{\tilde x},P_{\tilde y})$. We therefore have:
\begin{eqnarray*}
\int_{\tilde X\times Z}U(\tilde x,z)d\mu_1(\tilde x,z)-\int_{\tilde Y\times Z}C(\tilde y,z)d\mu_2(\tilde y,z)
&=& \int_{\tilde X\times \tilde Y\times Z}\left[ U(\tilde x,z)-C(\tilde y,z) \right] d\gamma(\tilde x,\tilde y,z) \\
&\leq& \int_{\tilde X\times \tilde Y\times Z}\sup_{z\in Z}\left[ U(\tilde x,z)-C(\tilde y,z) \right] d\gamma(\tilde x,\tilde y,z) \\
&=&\int_{\tilde X\times \tilde Y}\sup_{z\in Z}\left[ U(\tilde x,z)-C(\tilde y,z) \right] d\mu(\tilde x,\tilde y).
\end{eqnarray*}
Since $\nu$, $\mu_1$ and $\mu_2$ are arbitrary, the latter sequence of displays shows that the value of program (\ref{equilibrium1}) is smaller than or equal to the value of program~(\ref{equilibrium2}) below.
\begin{eqnarray}
\label{equilibrium2}
\sup_{\mu\in\mathcal M(P_{\tilde x},P_{\tilde y})}\int_{\tilde X\times \tilde Y}\sup_{z\in Z}\left[ U(\tilde x,z)-C(\tilde y,z) \right] d\mu(\tilde x,\tilde y).
\end{eqnarray}
By Assumption~\ref{ass:ec}, Program (\ref{equilibrium2}) is attained by the projection $\bar\mu$ on $\tilde X\times\tilde Y$ of the equilibrium probability distribution $\gamma$.
Let $\mu_1$ and $\mu_2$ be the  projections of the equilibrium probability distribution $\gamma$ onto $\tilde X \times Z$ and  $\tilde Y \times Z$.
Now, since the value of (\ref{equilibrium1}) is smaller than or equal to the value of (\ref{equilibrium2}), and the latter is attained by $\bar\mu$, we have
\begin{eqnarray*}
(\ref{equilibrium1})
&\leq& \int_{\tilde X\times \tilde Y}\sup_{z\in Z}\left[ U(\tilde x,z)-C(\tilde y,z) \right] d\bar\mu(\tilde x,\tilde y)\\
&=& \int_{\tilde X\times Z} U(\tilde x,z)d\mu_1(\tilde x,z)+\int_{\tilde Y\times Z}\left[ -C(\tilde y,z)\right] d\mu_2(\tilde y,z)\\
&\leq& \sup_{\mu_1\in\mathcal M(P_{\tilde x},P_z)}\int_{\tilde X\times Z} U(\tilde x,z)d\mu_1(\tilde x,z)+\sup_{\mu_2\in\mathcal M(P_{\tilde y},P_z)}\int_{\tilde Y\times Z}\left[ -C(\tilde y,z)\right] d\mu_2(\tilde y,z)
\leq (\ref{equilibrium1}),
\end{eqnarray*}
so that equality holds throughout, and $P_z$ solves (\ref{equilibrium1}).

\underline{Step 2:} We then show by contradiction that for $P_x$ almost every $x$, $P_{z|x}$ achieves
\begin{eqnarray}\label{equilibrium}
\max_\nu [T_U(P_{\varepsilon\vert x},\nu) +T_{-C}(P_{\tilde y|x},\nu)].
\end{eqnarray}
Assume the conclusion is false; that is, that there is some set $E \subset X$,  with $P_x(E) >0$ such that, for each $x$ in $E$  there is some probability distribution $P_{z|x}'$ on $Z$ such that
$$
 T_U(P_{\varepsilon| x}, P_{z|x}') + T_{-C}(P_{\tilde y|x},P_{z|x}') > T_U(P_{\varepsilon| x}, P_{z|x}) + T_{-C}(P_{\tilde y|x},P_{z|x}).
$$
For $x \notin E$, we set $P_{z|x}' = P_{z|x}$. This means that for every $x$, $P_{z|x}'$ is defined and we have
$$
 T_U(P_{\varepsilon| x}, P_{z|x}') + T_{-C}(P_{\tilde y|x},P_{z|x}') \geq  T_U(P_{\varepsilon| x}, P_{z|x}) + T_{-C}(P_{\tilde y|x},P_{z|x});
$$
moreover, the inequality is strict on a set $E$ of positive $P_x$ measure, so that
\begin{eqnarray}\label{eqn: surplus improvement}
\int_X[ T_U(P_{\varepsilon| x}, P_{z|x}') + T_{-C}(P_{\tilde y|x},P_{z|x}') ]dP_x>\int_X[T_U(P_{\varepsilon| x}, P_{z|x}) + T_{-C}(P_{\tilde y|x},P_{z|x})]dP_x.
\end{eqnarray}
We define a new probability distribution $P_z'$ on $Z$ by
$P_z'(A) =\int_{X} P_{z|x}'(A) dP_{x}(x),$
for each Borel set $A \subseteq Z$.  We claim that
$$
T_U(P_{\tilde x},P_z') + T_{-C}(P_{\tilde y},P_z') >T_U(P_{\tilde x},P_z) + T_{-C}(P_{\tilde y},P_z) ,
$$
which contradicts the maximality of $P_z$ in \eqref{equilibrium1} established in Step~1.
Let $P_{\varepsilon z|x}'$ achieve the optimal transportation between $P_{\varepsilon| x}$ and $P'_{z|x}$; that is,
$$
T_U(P_{\varepsilon| x}, P_{z|x}')  =\int_{\mathbb{R}^{d_{\varepsilon}} \times Z} U(x,\varepsilon, z)dP_{\varepsilon z|x}.
$$
Existence of such a $P_{\varepsilon z|x}$ is given by Theorem~1.3 page 19 of \citet{Villani:2003}
.
By construction, the marginals of~$P_{\varepsilon z|x}'$ are~$P_{\varepsilon| x}$ and~$P'_{z|x}$.
We define $P'_{\tilde xz}$ on $\tilde X \times Z$ for any Borel subset $B\subseteq \tilde X\times Z$ by $P'_{\tilde xz} (B):=\int_X P'_{\varepsilon z|x}(B_x)dP_x$,
where we define~$B_x:=\{(\varepsilon,z)\in\mathbb R^{d_z}\times Z:\;(x,\varepsilon,z)\in B\}$.
Then, for any Borel subset $A \subseteq Z$,
$$
P'_{\tilde xz}(\mathbb{R}^{d_z} \times X \times A) =\int_X P'_{\varepsilon z|x}(\mathbb{R}^{d_z} \times A)dP_x=\int_X P'_{z|x}(A)dP_x=P'_z(A).
$$
Similarly, for any Borel subset $B \subseteq \mathbb{R}^{d_z} \times X$,
$$
P'_{\tilde xz}(B \times Z) =\int_X P'_{\varepsilon z|x}(B_x \times Z)dP_x = \int_X P_{\varepsilon|x}(B_x)dP_x =P_{\tilde x} (B).$$
The last two calculations show that $P_{\tilde x}$ and $P'_z$ are the $\tilde x$ and $z$ marginals of $P'_{\tilde xz}$, respectively, and so it follows by definition that
\begin{eqnarray}\label{eqn: x,z optimality}
T_U(P_{\tilde x}, P'_z) \geq  \int_{ X\times\mathbb{R}^{d_z} \times Z} U(\tilde x, z)dP'_{\tilde xz}.
\end{eqnarray}
Similarly, we  let $P'_{\tilde yz|x}$ achieve the optimal transportation between $P_{\tilde y|x}$ and $P'_{z|x}$ and define $P'_{\tilde yz}$ on $\tilde Y \times Z$ by $P'_{\tilde yz} =\int_X P'_{\tilde yz|x}dP_x$;  a similar argument to above shows that the $\tilde y$ and $z$ marginals of $P'_{\tilde yz} $ are $P_{\tilde y}$ and $P'_z$, respectively.  Therefore,
\begin{eqnarray}\label{eqn: y,z optimality}
 T_{-C}(P_{\tilde y}, P'_z) \geq \int_{\tilde Y  \times Z} \left[ -C(\tilde y,z) \right] dP'_{\tilde yz}.
\end{eqnarray}

Now, equations \eqref{eqn: x,z optimality} and \eqref{eqn: y,z optimality} yield
\begin{eqnarray*}
T_U(P_{\tilde x}, P'_z) +T_{-C}(P_{\tilde y}, P'_z) &\geq&
\int_{ \tilde X \times Z} U(\tilde x, z)dP'_{\tilde xz}+\int_{\tilde Y  \times Z} \left[ -C(\tilde y, z) \right] dP'_{\tilde yz}\\
&=&\int_X\left[\int_{ \mathbb{R}^{d_z} \times Z} U(x,\varepsilon, z)dP'_{\varepsilon z|x}\right]dP_x + \int_X\left[\int_{\tilde Y  \times Z} \left[ -C(\tilde y, z) \right] dP'_{\tilde yz|x}\right]dP_x\\
&=&\int_X  T_U(P_{\varepsilon |x}, P'_{z|x})dP_x + \int_X[ T_{-C}(P_{\tilde y|x}, P_{z|x}')dP_x\\
&>&\int_X  T_U(P_{\varepsilon |x}, P_{z|x})dP_x + \int_X[ T_{-C}(P_{\tilde y|x}, P_{z|x})dP_x\\
&\geq &\int_X\left[\int_{ \mathbb{R}^{d_z} \times Z} U(x,\varepsilon, z)dP_{\varepsilon z|x}\right]dP_x + \int_X\left[\int_{\tilde Y  \times Z} \left[ -C(\tilde y, z) \right] dP_{\tilde yz|x}\right]dP_x\\
&=&\int_{ \tilde X \times Z} U(\tilde x, z)dP_{\tilde xz} + \int_{\tilde Y  \times Z} \left[ -C(\tilde y, z) \right] dP_{\tilde yz}\\
&=&T_U(P_{\tilde x}, P_z) +T_{-C}(P_{\tilde y}, P_z),
\end{eqnarray*}
where we have used \eqref{eqn: surplus improvement} in the fourth line above.  This establishes the contradiction and completes the proof.

\underline{Step 3:} We finally now show that for $P_x$ almost all $x$, $P_{z\vert x}$ is the unique solution of Program~(\ref{equilibrium}). This step follows the method of proof of Proposition~4 in \citet{CE:2010}.
We know from Step~2 that~$P_{z\vert x}$ is a solution to~(\ref{equilibrium}). Suppose~$\nu\in\Delta(Z)$ also solves~(\ref{equilibrium}). Let~$\mu_{\varepsilon z\vert x}$ and~$\mu_{\tilde yz\vert x}$ achieve~$T_U(P_{\varepsilon \vert x},\nu)$ and~$T_{-C}(P_{\tilde y\vert x},\nu)$. The latter exist by Theorem~1.3 page 19 of \citet{Villani:2003}
.
We therefore have
\begin{eqnarray*}
\int_{\mathbb R^{d_z}\times Z}U(x,\varepsilon,z)d\mu_{\varepsilon z\vert x}(\varepsilon,z\vert x)+\int_{\tilde Y\times Z}[-C(\tilde y,z)]d\mu_{\tilde yz\vert x}(\tilde y,z\vert x)=T_U(P_{\varepsilon\vert x},\nu)+T_{-C}(P_{\tilde y\vert x},\nu)=(\ref{equilibrium}).
\end{eqnarray*}
Let $\varphi_x^U$ denote the~$U$-conjugate of a function~$\varphi$ on~$Z$, as defined
by~$\varphi_x^U(\varepsilon):=\sup_{z\in Z}\{U(x,\varepsilon,z)-\varphi(z)\}$. Similarly, let~$(-\varphi)^{-C}$ be the~$(-C)$-conjugate of~$(-\varphi)$, defined as~$(-\varphi)^{-C}(\tilde y):=\sup_{z\in Z}\{-C(\tilde y,z)+\varphi(z)\}$. By definition, for any function~$\varphi$ on~$Z$, and any~$(\varepsilon,\tilde y,z)$ on~$\mathbb R^{d_z}\times\tilde Y\times Z$, we have
\begin{eqnarray}\label{eq:young}
\varphi_x^U(\varepsilon)+\varphi(z)\geq U(x,\varepsilon,z),\mbox{ and }[-\varphi]^{-C}(\tilde y)-\varphi(z)\geq -C(\tilde y,z).
\end{eqnarray}
By Assumption \ref{ass:ec}, the price function $p$ of the hedonic equilibrium satisfies
\begin{eqnarray}\label{eq:ec1}
p_x^U(\varepsilon)+p(z)= U(x,\varepsilon,z), P_{\varepsilon z\vert x}\mbox{ a.s. and }[-p]^{-C}(\tilde y)-p(z)= -C(\tilde y,z), P_{\tilde yz\vert x}\mbox{ almost surely}.
\end{eqnarray}
Therefore, $p$ achieves the minimum in the program
\begin{eqnarray}\label{equilibrium3}
&&\inf_{\varphi}\left\{ \int_{\mathbb R^{d_z}\times Z}[\varphi_x^U(\varepsilon)+\varphi(z)]dP_{\varepsilon z\vert x}(\varepsilon, z\vert x)+\int_{\tilde Y\times Z} [(-\varphi)^{-C}(\tilde y)-\varphi(z)]dP_{\tilde yz\vert x}(\tilde y,z\vert x) \right\}\nonumber\\
&=&\inf_{\varphi}\left\{ \int_{\mathbb R^{d_z}}\varphi_x^U(\varepsilon)dP_{\varepsilon\vert x}(\varepsilon\vert x)+\int_Z\varphi(z)dP_{z\vert x}(z\vert x)+\int_{\tilde Y} (-\varphi)^{-C}(\tilde y)dP_{\tilde y\vert x}(\tilde y\vert x)-\int_Z\varphi(z)dP_{z\vert x}(z\vert x) \right\}\nonumber\\
&=&\inf_\varphi\left\{ \int_{\mathbb R^{d_z}}\varphi_x^U(\varepsilon)dP_{\varepsilon\vert x}(\varepsilon\vert x)+\int_{\tilde Y} (-\varphi)^{-C}(\tilde y)dP_{\tilde y\vert x}(\tilde y\vert x) \right\}.
\end{eqnarray}
Since $p$ achieves (\ref{equilibrium3}), we have
\begin{eqnarray*}
(\ref{equilibrium3})&=&\int_{\mathbb R^{d_z}}p_x^U(\varepsilon)dP_{\varepsilon\vert x}(\varepsilon\vert x)+\int_{\tilde Y} (-p)^{-C}(\tilde y)dP_{\tilde y\vert x}(\tilde y\vert x) \\
&=&\int_{\mathbb R^{d_z}\times Z}[p_x^U(\varepsilon)+p(z)]d\mu_{\varepsilon z\vert x}(\varepsilon, z\vert x)+\int_{\tilde Y\times Z} [(-p)^{-C}(\tilde y)-p(z)]d\mu_{\tilde yz\vert x}(\tilde y,z\vert x),
\end{eqnarray*}
since $\mu_{\varepsilon z\vert x}\in\mathcal M(P_{\varepsilon\vert x},\nu)$ and $\mu_{\tilde yz\vert x}\in\mathcal M(P_{\tilde y\vert x},\nu)$ by construction. By the strong duality result in Theorem~3 of \citet{CE:2010}, we have (\ref{equilibrium})=(\ref{equilibrium3}). Hence
\begin{eqnarray*}
&& \int_{\mathbb R^{d_z}\times Z}U(x,\varepsilon,z)d\mu_{\varepsilon z\vert x}(\varepsilon,z\vert x)+\int_{\tilde Y\times Z}[-C(\tilde y,z)]d\mu_{\tilde yz\vert x}(\tilde y,z\vert x)\\
&&\hskip100pt=
\int_{\mathbb R^{d_z}\times Z}[p_x^U(\varepsilon)+p(z)]d\mu_{\varepsilon z\vert x}(\varepsilon, z\vert x)+\int_{\tilde Y\times Z} [(-p)^{-C}(\tilde y)-p(z)]d\mu_{\tilde yz\vert x}(\tilde y,z\vert x),
\end{eqnarray*}
which, given the inequalities (\ref{eq:young}), implies
\begin{eqnarray*}
p_x^U(\varepsilon)+p(z)= U(x,\varepsilon,z), \mu_{\varepsilon z\vert x}\mbox{ almost surely, and }[-p]^{-C}(\tilde y)-p(z)= -C(\tilde y,z), \mu_{\tilde yz\vert x}\mbox{ almost surely}.
\end{eqnarray*}
We therefore have $p_x^U(\varepsilon)+p(z)= U(x,\varepsilon,z)$ both $P_{\varepsilon z\vert x}$ and $\mu_{\varepsilon z\vert x}$ almost surely.
The $U$-conjugate $p_x^U$ of $p$ is locally Lipschitz by Lemma~C.1 of \citet{GM:96}, hence differentiable $P_{\varepsilon|x}$ almost everywhere by Rademacher's Theorem (see for instance \citet{Villani:2009}, Theorem~10.8(ii)). We can therefore apply the Envelope Theorem to obtain the equation below.
\begin{eqnarray}
\label{eq:rad}
\nabla p_x^U(\varepsilon)=\nabla_\varepsilon U(x,\varepsilon,z)=\nabla_\varepsilon\zeta(x,\varepsilon,z),
\mbox{ both }P_{\varepsilon z|x}\mbox{ and }\mu_{\varepsilon z|x}\mbox{ almost everywhere}.
\end{eqnarray}
By Assumption~\ref{ass:twist}(B), $\nabla_\varepsilon\zeta(x,\varepsilon,z)$ is injective as a function of $z$. Therefore, for each $\varepsilon$, there is a unique~$z$ that satisfies (\ref{eq:rad}). This defines a map $T:\mathbb R^{d_z}\rightarrow Z$, such that $z=T(\varepsilon)$ both $P_{\varepsilon z|x}$ and $\mu_{\varepsilon z|x}$ almost everywhere. Since the projections of $P_{\varepsilon z|x}$ and $\mu_{\varepsilon z|x}$ with respect to $\varepsilon$ are the same, namely $P_{\varepsilon|x}$, we therefore have
\[
P_{\varepsilon z|x}=(id,T)\#P_{\varepsilon|x}=\mu_{\varepsilon z|x}.
\]
The two probability distributions $P_{\varepsilon z|x}$ and $\mu_{\varepsilon z|x}$ being equal, they must also share the same projection with respect to $z$ and $\nu=P_{z|x}$ as a result.

Armed with the results in Steps~1 to~3, we are ready to prove Lemma~\ref{lemma:abs3}.
Fix~$x\in X$ such that~$P_{z|x}$ is the unique solution to Program~\eqref{equilibrium}. Let~$P_{\tilde y}^N$ be a sequence of discrete probability distributions with~$N$ points of support on~$\tilde Y \subseteq\mathbb{R}^{d_{\tilde y}}$ converging weakly to~$P_{\tilde y|x}$.
The set of probability distributions on the compact set~$Z$ is compact relative to the topology of weak convergence. By Assumption~\ref{ass:supp3},~$U$ and~$C$ are continuous, hence, by Theorem~5.20 of \citet{Villani:2009}, the functional~$\nu\mapsto T_U(P_{\varepsilon\vert x},\nu) +T_{-C}(P_{\tilde y}^N,\nu)$ is continuous with respect to the topology of weak convergence. Program $\max_{\nu\in\Delta(Z)}[ T_U(P_{\varepsilon\vert x},\nu) +T_{-C}(P_{\tilde y}^N,\nu)],$ therefore, has a solution we denote~$P_z^N$.

We first show that $P_z^N$ converges weakly to $P_{z|x}$.
For any probability measure $\nu$ on $Z$, we have
\begin{eqnarray}\label{eq:sub}
T_U(P_{\varepsilon\vert x},\nu) +T_{-C}(P_{\tilde y}^N,\nu) \leq T_U(P_{\varepsilon\vert x},P_z^N) +T_{-C}(P_{\tilde y}^N,P_z^N).
\end{eqnarray}
Since $Z$ is compact, $\Delta(Z)$ is compact with respect to weak convergence. Hence, we can extract from $(P_z^N)$ a convergent subsequence, which we also denote~$(P_z^N)$, as is customary, and we call the limit~$\bar P$. By the stability of optimal transport~(\citet{Villani:2009}, Theorem 5.20), we have~$T_{-C}(P_{\tilde y}^N,\nu)\rightarrow T_{-C}(P_{\tilde y|x},\nu)$, $T_U(P_{\varepsilon\vert x},P_z^N) \rightarrow T_U(P_{\varepsilon\vert x},\bar P)$  and~$T_{-C}(P_{\tilde y}^N, \bar P) \rightarrow T_{-C}(P_{\tilde y|x}, \bar P) $, and so passing to the limit in inequality~\eqref{eq:sub} yields
 $$
T_U(P_{\varepsilon\vert x},\nu) +T_{-C}(P_{\tilde y|x},\nu) \leq T_U(P_{\varepsilon\vert x},\bar P) +T_{-C}(P_{\tilde y|x} , \bar P).
$$
As this holds for any probability distribution~$\nu$ on~$\mathbb{R}^{d_z}$, it implies that~$\bar P$ is optimal in Program~\eqref{equilibrium}. By uniqueness proved in Step~3, we then have $\bar P = P_{z|x}$, as desired.

We are now ready to complete the proof of Lemma~\ref{lemma:abs3}. Combining Steps~1 to~3 above, we know that $P_z^N$ is the marginal with respect to $z$ of a hedonic equilibrium distribution $\gamma^N$ on $\mathbb R^{d_z}\times Z\times\{\tilde y_1,\ldots,\tilde y_N\}$, with consumer and producer distributions $P_{\varepsilon\vert x}$ and $P_{\tilde y}^N$, respectively. By Step 1 in the proof of Theorem~1(1) in \citet{CGHP:2020} applied to this hedonic equilibrium with producer type distribution $P_{\tilde y}^N$, for each $N$, there is an optimal map~$F_N$ pushing $P_z^N$ forward to $P_{\varepsilon\vert x}$ i.e., $P_{\varepsilon\vert x}=F_N\#P_z^N$.
For $i=1,2,...N$, we define the subsets~$S_i^N \subseteq Z$ by
$$
S_i ^N=\{ z\in Z:\;z\in\mbox{arg}\max_w(U(x,F_N(z), w) - C(\tilde y_i, w)) \};
$$
note that $S_i^N$ is the set of quality vectors that are produced by producer type $\tilde y_i$.  Since $P_{\tilde y}^N$ has finite support~$\{\tilde y_1,\ldots,\tilde y_N\}$, $P_z^N$ almost all $z$ belong to some $S_i$, $i=1,\ldots,N$. We then set
$$
E^N_i =\{ z \in S_i^N \text{ and } z \notin S_j^N \text{ for all } j <i\},
$$
for each $i=1,\ldots,N$, with the convention $S_0^N=\varnothing$.
The $E^N_i$ are disjoint, and $E^N=\cup_{i=1}^N E^N_i$ has full $P_z^N$ measure, $P_z^N(E^N)=1$.  On each $E^N_i$,  $F_N$ coincides with a map $G_i$, which satisfies
$$
\nabla_zU(x,G_i(z),z) -\nabla_zC(\tilde y_i, z)=0
$$
or
$$
\nabla_z\bar U(x,z) + \nabla_z\zeta(x,G_i(z),z) -\nabla_zC(\tilde y_i, z)=0.
$$
By Assumption~\ref{ass:supp3}(3) and the Implicit Function Theorem, $G_i$ is differentiable and we have
$$
[D^2_{zz}\bar U(x,z) + D^2_{zz}\zeta(x,G_i(z),z) -D^2_{zz}C(\tilde y_i, z)] +D^2_{z\varepsilon}\zeta(x,G_i(z),z)DG_i(z)=0
$$
so that
\begin{eqnarray*}
DG_i(z)&=& [D^2_{z\varepsilon}\zeta(x,G_i(z),z)]^{-1}[D^2_{zz}\bar U(x,z) + D^2_{zz}\zeta(x,G_i(z),z) -D^2_{zz}C(\tilde y_i, z)].
\end{eqnarray*}
Therefore, $\Vert DG_i(z)\Vert \leq  M_0M_1:=C$ by Assumption~\ref{ass:supp3}.   Now, this implies that~$G_i$ is Lipschitz with constant~$C$, and therefore,~$F_N$ restricted to~$E_i^N$ is also Lipschitz with constant~$C$.

Now, for any Borel $A \subset \mathbb{R}^{d_z}$, we can write $A \cap E^N =\cup_{i=1}^N(A\cap E^N_i) $.  Therefore,
\begin{eqnarray}\label{pushforward}
P_{z}^N(A) =P_{z}^N(A\cap E^N )&\leq& P_{z}^N(F_N^{-1}(F_N(A\cap E^N )))\nonumber\\
&=&P_{\varepsilon\vert x}( F_N(A\cap E^N)) \nonumber\\
&=& P_{\varepsilon\vert x}(F_N(\cup_{i=1}^N(A\cap E^N_i) )\nonumber\\
&=&P_{\varepsilon\vert x}(\cup_{i=1}^NF_N(A\cap E^N_i) ).
\end{eqnarray}

Denote by $|A|$ the Lebesgue measure of a set $A$ in the rest of this proof. We now show the absolute continuity of of $P_{z|x}$ by contradiction.  Assume not; then there is a set $A$ with $|A|=0$  but $\delta=P_{z|x}(A) > 0$.  We can choose open neighbourhoods, $A \subseteq A_k $, with $|A_k| \leq \frac{1}{k}$;  by weak convergence of $P_z^N$ to $P_{z|x}$, we have
$$
\liminf_{N \rightarrow\infty} P_z^N(A_k) \geq P_{z|x}(A_k) \geq \delta
$$
and so for sufficiently large $N$, we have
$P_z^N(A_k) \geq \delta/2.$
On the other hand, by the Lipschitz property of $F_N$ on $E_i^N$, we have
$$|F_N(A_k\cap E_i^N)|  \leq C^{d_z}|A_k\cap E_i^N|,$$
so that
$$
|\cup_{i=1}^NF_N(A_k\cap E_i^N)|  \leq C^{d_z}\sum_{i=1}^N|A_k\cap E_i^N| =C^{d_z}|A_k \cap E^N| \leq C^{d_z}|A_k| \leq\frac{C^{d_z}}{k}.
$$
Now,   $P_{\varepsilon\vert x} $ is absolutely continuous, so that
$
P_{\varepsilon\vert x}(\cup_{i=1}^NF_N(A_k\cap E_i^N)) \rightarrow 0
$
as $k \rightarrow \infty$ (because $|\cup_{i=1}^NF_N(A_k\cap E_i^N)| \rightarrow 0$ as $k \rightarrow \infty$).
On the other hand, by \eqref{pushforward}, $P_{\varepsilon\vert x}(\cup_{i=1}^NF_N(A_k\cap E_i^N)) \geq P_z^N(A_k) \geq \frac{\delta}{2}$ for all $k$, which is a contradiction, completing the proof.
\hfill$\square$

}

\subsection{Smoothness of the endogenous price function}

\begin{definition}[Subdifferential]\label{def:subdif}
The {\em subdifferential} $\partial\psi(x_0)$ of a function $\psi:\mathbb R^d\rightarrow\mathbb R\cup\{+\infty\}$ at $x_0\in\mathbb R^d$ is the set of vectors $p\in\mathbb R^d$, called {\em subgradients}, such that $\psi(x)-\psi(x_0)\geq p'(x-x_0)+o(\|x-x_0\|)$.
\end{definition}

\begin{definition}[$\zeta$-conjugate]\label{def:zconj}
Let~$\zeta$ be a function on~$\mathbb R^{d_\varepsilon}\times\mathbb R^{d_z}$ that is continuously differentiable and satisfies Assumption~\ref{ass:twist}. The {\em $\zeta$-conjugate} $\psi^\zeta$ of a function $\psi:\mathbb R^d\rightarrow\mathbb R\cup\{+\infty\}$ is $\psi^\zeta(\varepsilon)=\sup_{z}\{\zeta(\varepsilon,z)-\psi(z)\}$.
\end{definition}

\begin{definition}[$\zeta$-subdifferential]\label{def:zsubdif}
Let~$\zeta$ be a function on~$\mathbb R^{d_\varepsilon}\times\mathbb R^{d_z}$ that is continuously differentiable and satisfies Assumption~\ref{ass:twist}. The {\em $\zeta$-subdifferential} $\partial^{\zeta}\psi(z)$ of a function $\psi:\mathbb R^d\rightarrow\mathbb R\cup\{+\infty\}$ at $z\in\mathbb R^d$ is the set of vectors $\varepsilon\in\mathbb R^d_{\varepsilon}$, such that $\psi^\zeta(\varepsilon)+\psi(z)=\zeta(\varepsilon,z)$.
\end{definition}

\begin{definition}[Local semiconvexity]
\label{def:sc}
A function $\psi:\mathbb R^d\rightarrow\mathbb R\cup\{+\infty\}$ is called {\em locally semiconvex} at~$x_0\in\mathbb R^d$ if there is a scalar $\lambda>0$ such that $\psi(x)+\lambda \|x\|^2$ is convex on some open ball centered at $x_0$.
\end{definition}

Since the term $\lambda\|x\|^2$ in the definition of local semiconvexity simply shifts the subdifferential by $2\lambda x$, we can extend Theorem 25.6 in \citet{Rockafellar:70} to locally semiconvex functions and obtain the following lemma.
\begin{lemma}
\label{lem:sconv}
Let~$\psi:\mathbb R^d\rightarrow\mathbb R\cup\{+\infty\}$ be a locally semiconvex function, and suppose that~$q\in\mathbb R^d$ is an extremal point in the subdifferential~$\partial \psi(x_0)$ of~$\psi$ at~$x_0$.  Then there exists a sequence~$x_n$ converging to~$x_0$, such that~$\psi$ is differentiable at each~$x_n$ and the gradient~$\nabla \psi(x_n)$ converges to~$q$.
\end{lemma}

\begin{definition}[Approximate differentiability]\label{def:apd}
If we have
\[
\lim_{r\downarrow0}\frac{|B_r(x)\cap E|}{|B_r(x)|}=1,
\]
where $|B|$ is Lebesgue measure of $B$, $B_r(x)$ is an open ball of radius $r$ centered at $x$, and $E\subseteq \mathbb R^d$ is a measurable set,
$x$ is called {\em density point} of $E$.
Let $x_0$ be a density point of a measurable set $E\subseteq\mathbb R^d$ and let~$f:E\rightarrow\mathbb R$ be a measurable map. If there is a linear map $A: \mathbb R^d\rightarrow\mathbb R$ such that, for each $\eta>0$, $x_0$ is a density point of
\[
\left\{ x\in E:\; -\eta \leq \frac{f(x)-f(x_0)-A(x-x_0)}{\| x-x_0\|} \leq \eta \right\},
\]
then $f$ is said to be {\em approximately differentiable at} $x_0$, and $A$ is called the {\em approximate gradient} of $f$ at $x_0$.
The approximate gradient is uniquely defined, as shown below Definition~10.2 page 218 of \citet{Villani:2009}.
\end{definition}

\begin{lemma}[Approximate differentiability of the potential]
\label{lemma:ad}
Under Assumptions~\ref{ass:ec}, ~\ref{ass:uh'},~\ref{ass:twist} and~\ref{ass:supp3}, the potential~$z\mapsto V(x,z)$ is approximately differentiable~$P_{z\vert x}$ almost everywhere.
\end{lemma}

\subsubsection*{Proof of Lemma~\ref{lemma:ad}}

{\footnotesize

By Lemma~\ref{lemma:zconv}, we know that for each~$x$,~$V(x,z)=V^{\zeta\zeta}(x,z)$,~$P_{z|x}$ almost everywhere, and~$P_{z|x}$  is absolutely continuous with respect to Lebesgue measure, by Lemma~\ref{lemma:abs3}.  We will prove that~$V$ is therefore approximately differentiable~$P_{z|x}$ almost everywhere, with~$\nabla_{ap,z}V(x,z)=\nabla_z V^{\zeta\zeta}(x,z),$
where~$\nabla_{ap,z}V(x,z)$ denotes the approximate gradient of~$V(x,z)$ with respect to~$z$.
By Lemma~\ref{lemma:zconv}, we have $V=V^{\zeta\zeta}$, $P_{z|x}$ almost everywhere.  Moreover, as a $\zeta$ conjugate, $V^{\zeta\zeta}$ is locally Lipschitz by Lemma~C.1 of \citet{GM:96}, hence differentiable $P_{z|x}$ almost everywhere by Rademacher's Theorem (see for instance \citet{Villani:2009}, Theorem~10.8(ii)). Hence, there exists a set $S$ of full $P_{z|x}$  measure such that, for each $z_0 \in S$,
\begin{enumerate}
\item $V^{\zeta\zeta}(x,z)$ is differentiable with respect to $z$ at $z_0$.
\item $V(x,z_0)=V^{\zeta\zeta}(x,z_0)$.
\end{enumerate}
By Lebesgue's density theorem (still denoting Lebesgue measure of $B$ by $|B|$),
$$
\lim_{r\downarrow 0} \frac{|S \cap B_r(z)|}{|B_r(z)|} =1
$$
for Lebesgue almost every $z\in Z$, hence also for $P_{z\vert x}$ almost every $z$, by the absolute continuity of $P_{z|x}$. Since $S$ has $P_{z\vert x}$ measure~$1$,
the set $\bar S$ of density points of $S$,
$$
\bar S =\left\{z \in S: \lim_{r\downarrow 0} \frac{|S \cap B_r(z)|}{|B_r(z)|} =1\right\},
$$
therefore has $P_{z|x}$ measure~$1$.  Take any density point $z_0$ of $S$, i.e., $z_0 \in \bar S$.
Fix $\eta>0$. Since $z_0\in\bar S\subseteq S$, $V^{\zeta\zeta}$ is differentiable at $z_0$. Hence there is $r>0$ such that for all $z\in B_r(z_0)$,
\[
-\eta\leq \frac{V^{\zeta\zeta}(z)-V^{\zeta\zeta}(z_0)-\nabla_zV^{\zeta\zeta}(z_0)\cdot(z-z_0)}{\|z-z_0\|} \leq\eta.
\]
Since $V^{\zeta\zeta}=V$ on $S$, for all $z\in B_r(z_0)\cap S$, we have
\[
-\eta\leq \frac{V(z)-V(z_0)-\nabla_zV^{\zeta\zeta}(z_0)\cdot(z-z_0)}{\|z-z_0\|} \leq\eta.
\]
Therefore $B_r(z_0)\cap S=B_r(z_0)\cap \tilde S$, where
\[
\tilde S:=\left\{z\in S: -\eta\leq \frac{V(z)-V(z_0)-\nabla_zV^{\zeta\zeta}(z_0)\cdot(z-z_0)}{\|z-z_0\|} \leq\eta\right\}.
\]
As $z_0$ is a density point of $S$,
\[
\lim_{r\downarrow 0} \frac{|\tilde S \cap B_r(z_0)|}{|B_r(z_0)|} =\lim_{r\downarrow 0} \frac{|S \cap B_r(z_0)|}{|B_r(z_0)|} =1,
\]
so that $z_0$ is also a density point of $\tilde S$,
which means, by definition, that $V$ is approximately differentiable at~$z_0$, and that its approximate gradient is
$\nabla_{ap,z} V(x,z_0)=\nabla_z V^{\zeta\zeta}(x,z_0).$ The latter is true for any $z_0\in\bar S$ and $\bar S$ has $P_{z\vert x}$ measure $1$. Hence, $V$ is approximately differentiable $P_{z\vert x}$ almost everywhere.
\hfill%
$\square$

}

\bibliographystyle{abbrvnat}
\bibliography{HedonicsCGHP}

\end{document}